\documentclass[11pt]{article}
\usepackage{jheppub}
\usepackage{graphicx} 
\usepackage{graphbox}
\usepackage[usenames,dvipsnames]{xcolor} 
\usepackage{tikz}	
\usetikzlibrary{matrix}
\usepackage{dirtytalk}
\usepackage{dsfont}
\usepackage{colortbl}
\usepackage{braket}
\usepackage{amsthm}
\usepackage{amsmath}
\usepackage{float}
\usepackage{bbm, dsfont}
\usepackage[normalem]{ulem}
\usepackage{slashed}
\usepackage{mathtools}
\usepackage{multirow}
\usepackage{cancel}
\usepackage{cleveref}
\usepackage{orcidlink}
\usepackage{enumerate}
\usepackage{fontawesome}
\usepackage{adjustbox}
\usepackage{tcolorbox}
\usepackage{bm}

\newcommand{\bs}[1]{\boldsymbol{#1}}
\newcommand{\rd}{\mathrm{d}}

\newcommand{\cA}{\mathcal{A}}

\DeclareMathOperator{\Const}{Const}
\DeclareMathOperator{\Frac}{Frac}
\DeclareMathOperator{\SL}{SL}
\DeclareMathOperator{\K}{K}
\DeclareMathOperator{\E}{E}
\DeclareMathOperator{\lm}{lm}
\DeclareMathOperator{\Ker}{Ker}
\newcommand{\cF}{\mathcal{F}_{\mathbb{C}}}
\newcommand{\ccF}{\mathcal{F}}

\newcommand{\bx}{\bs{x}}
\newcommand{\KK}{\mathbb{K}}

\newcommand{\bI}{\bs{I}}
\newcommand{\bP}{\bs{P}}
\newcommand{\bA}{\bs{A}}

\newcommand{\bS}{\bs{S}}
\newcommand{\cS}{\mathcal{S}}
\newcommand{\bF}{\bs{F}}
\newcommand{\bH}{\bs{H}}
\newcommand{\bC}{\bs{C}}
\newcommand{\bK}{\bs{K}}
\newcommand{\bU}{\bs{U}}
\newcommand{\bAc}{\bs{\check{A}}}
\newcommand{\bPc}{\bs{\check{P}}}
\newcommand{\bIc}{\bs{\check{I}}}
\newcommand{\bJc}{\bs{\check{J}}}
\newcommand{\bSc}{\bs{\check{S}}}
\newcommand{\bOmegac}{\bs{\check{\Omega}}}
\newcommand{\bR}{\bs{R}}

\newcommand{\bM}{\bs{M}}
\newcommand{\bJ}{\bs{J}}
\newcommand{\bG}{\bs{G}}

\newcommand{\bD}{\bs{D}}
\newcommand{\bOmega}{\bs{\Omega}}
\newcommand{\bDelta}{\bs{\Delta}}
\newcommand{\balpha}{\bs{\alpha}}

\newcommand{\bnu}{\bs{\nu}}
\newcommand{\cB}{\mathcal{B}}
\newcommand{\bbA}{\mathbb{A}}

\DeclareMathOperator{\dlog}{\rd\!\log}

\def\beq{\begin{equation}}
\def\eeq{\end{equation}}
\def\bsp#1\esp{\begin{split}#1\end{split}}

\newtheorem{definition}{Definition}
\newtheorem{theorem}{Theorem}

\newtheorem{lemma}{Lemma}
\newtheorem{corollary}{Corollary}
\theoremstyle{definition}

\newcounter{examplenew} 
\newcounter{subexampleone}[examplenew]
\newcounter{subexampletwo}[examplenew]
\renewcommand{\thesubexampleone}{\arabic{examplenew}.\arabic{subexampleone}}
\renewcommand{\thesubexampletwo}{\arabic{examplenew}.\arabic{subexampletwo}}

\usepackage{xparse}

\NewDocumentEnvironment{examplenew}{O{} m}
{
    \ifnum#2=1
        \setcounter{examplenew}{1}
        \refstepcounter{subexampleone}
        \par \medskip \noindent \textbf{Example \thesubexampleone}%
    \else
        \ifnum#2=2
            \setcounter{examplenew}{2}
             \refstepcounter{subexampletwo}
           \par \medskip\noindent  \textbf{Example \thesubexampletwo}%
        \fi
    \fi
    \ifx&#1&%
    \else
        ~(#1)~%
    \fi
    \rmfamily 
}
{   \par
    \medskip
    \noindent}

\newcommand{\eps}{\varepsilon}

\DeclareMathOperator{\diag}{diag}
\DeclareMathOperator{\End}{End}
\DeclareMathOperator{\GL}{GL}

\title{Self-duality from twisted cohomology}
\abstract{Recently a notion of self-duality for differential equations of maximal cuts was introduced, which states that there should be a basis in which the matrix for an $\eps$-factorised differential equation is persymmetric. It was observed that the rotation to this special basis may introduce a Galois symmetry relating different integrals. We argue that the proposed notion of self-duality for maximal cuts stems from a very natural notion of self-duality from twisted cohomology. Our main result is that, if the differential equations and their duals are simultaneously brought into canonical form, the cohomology intersection matrix is a constant. Furthermore, we show that one can associate quite generically a Lie algebra representation to an $\eps$-factorised system. For maximal cuts, this representation is irreducible and self-dual. The constant intersection matrix can  be interpreted as expressing the equivalence of this representation and its dual, which in turn results in constraints for the differential equation matrix. Unlike the earlier proposal, the most natural symmetry of the differential equation matrix is defined entirely over the rational numbers and is independent of the basis choice. 
}
\author[a]{Claude Duhr,}
\emailAdd{cduhr@uni-bonn.de}
\affiliation[a]{Bethe Center for Theoretical Physics, Universität Bonn, D-53115, Germany
}

\author[a
]{Franziska Porkert,}
\emailAdd{fporkert@uni-bonn.de}

\author[a]{Cathrin Semper,}
\emailAdd{csemper@uni-bonn.de}

\author[a
]{Sven F. Stawinski,}
\emailAdd{sstawins@uni-bonn.de}

\begin{document}

\begin{flushright}
BONN-TH-2024-11
\end{flushright}

\definecolor{lightorange}{rgb}{1,0.6875,0.5}
\definecolor{lightblue}{rgb}{0.68, 0.85, 0.9}

\maketitle


\newpage
\section*{Index of notations}
\begin{center}
    \begin{tabular}{ll}
    $\mathbb{K}$ & an algebraic number field.\\
        $\mathbb{K}(\eps)$ & the field of rational functions in $\eps$ with coefficients in $\KK$.\\
                $R^{N\times N}$ &for any ring $R$, the algebra of $N\times N$ matrices with entries in $\KK$.\\
                $\GL(N,R)$ & the multiplicative subgroup of $R^{N\times N}$ of matrices with full rank.\\
    $\mathcal{A}$ & a differentially closed $\mathbb{K}$-algebra whose constants are in $\mathbb{K}$.\\
     $\Omega^1(\cA)$&the $\KK$-vector space of differential forms of the form $\sum_{k=1}^r\rd x_k\,f_{k}(\bx)$, $f_{k} \in \cA$.\\
 $Z^1(\cA) $ & the subspace of $\omega\in \Omega^1(\cA)$ with $\rd \omega=0$.\\
 $\mathbb{A}$& a finite-dimensional subspace of  $Z^1(\cA) $ with basis $\omega_1,\ldots,\omega_p$.\\
$\boldsymbol{A}(\boldsymbol{x})$& a matrix with entries in $\bbA$.\\
 $\cA_{\eps}$&$ \cA\otimes_{\KK} \KK(\eps).$\\
$\cF$&$\Frac(\mathbb{C}\otimes_{\KK}\cA).$\\
$A_{\mathbb{A},1}$& a $\mathbb{K}$-vector space generated by the letters $t_i$. \\ 
$\mathfrak{g}_{\bbA}$&the Lie algebra generated by the $t_i$ (with some specific commutation relations).\\
$A_{\mathbb{A}}$ & the universal enveloping algebra of $\mathfrak{g}_{\bbA}$.\\
$\mathbb{G}$& the generating series of all words in the letters $t_i$ with coefficients that are \\&iterated integrals of the forms $\omega_i$.
%
%
    \end{tabular}
\end{center}

\section{Introduction}
\label{sec:intro}

Feynman integrals are the foundation of essentially all higher-order computations in perturbative Quantum Field Theory (pQFT), and as such they are extremely important both as a tool to understand the mathematical structure underlying pQFT and to make precise predictions for collider and gravitational wave experiments. It is therefore important to have efficient tools for their evaluation. Experience from the past shows that the development of novel computational tools often goes hand-in-hand with an increased understanding of the underlying mathematics. For this reason, a lot of effort has recently gone into studying the mathematics underlying Feynman integrals.

Feynman integrals typically diverge, and need to be regulated. The most commonly used regularisation scheme is dimensional regularisation~\cite{THOOFT1972189}, where the integrals are computed in $D=4-2\eps$ dimensions, and expanded around $\eps=0$ at the very end. Dimensionally-regulated Feynman integrals are meromorphic function of $\eps$, and the divergencies show up as poles in $\eps$. Over the last decade, there was enormous progress in our understanding of dimensionally-regulated Feynman integrals (see, e.g., ref.~\cite{Bourjaily:2022bwx} for a recent review). One of the most important advances was the realisation that for every family of Feynman integrals, there is a distinguished set of bases, sometimes called a \emph{canonical} or \emph{$\eps$-factorised} bases, such that the differential equation takes a particularly simple form where the dimensional regulator $\eps$  enters only linearly~\cite{Henn:2013pwa}. Canonical bases are at the origin of almost all modern results for multi-loop Feynman integrals. The existence of canonical bases is still conjectural, but by now there is a massive body of evidence that such a basis exists at least for very large classes of Feynman integrals. Another important advance was the discovery that the appropriate mathematical framework to study dimensionally-regulated Feynman integrals is twisted cohomology~\cite{Mastrolia:2018uzb}, which has seen numerous applications to Feynman integral computations~\cite{Mizera:2017rqa,Mizera:2019gea,Mizera:2019ose,matsumoto_relative_2019-1,Abreu:2019wzk,Britto:2021prf,Caron-Huot:2021xqj,Caron-Huot:2021iev,Chen:2022lzr,Giroux:2022wav,Gasparotto:2023roh,Fontana:2023amt,Bhardwaj:2023vvm,De:2023xue,Chen:2023kgw,Duhr:2023bku,Brunello:2023rpq,Crisanti:2024onv}, most notably to identify novel ways to reduce Feynman integrals to a basis.

Despite all this progress, there are still many open questions related to multi-loop Feynman integrals, and their explicit evaluation still remains challenging, especially for integrals depending on many loops and scales. In order to make further progress, it can therefore be fruitful to focus on a simpler set of objects, which nevertheless retain the important mathematical features of Feynman integrals. Such a set of simpler objects are maximal cuts, which can be defined, loosely speaking, as the integrals obtained after taking residues at all propagator poles, which physically corresponds to putting all propagators on their mass shell. In addition, maximal cuts can have additional properties which may deserve attention in their own right. For example, when evaluated at $\eps=0$, maximal cuts often evaluate to periods of some algebraic varieties. They thus provide a direct link between algebraic geometry and pQFT. Moreover, twisted cohomology endows maximal cuts with a natural notion of self-duality, which leads to quadratic relations between different maximal cuts~\cite{Duhr:2024rxe}. 

Recently, another notion of self-duality was observed for explicit results for some maximal cuts~\cite{Pogel:2022ken,Pogel:2022vat,Pogel:2022yat,Pogel:2024sdi}. This apparently new concept of self-duality has a priori  no connection to twisted cohomology and states that there is a distinguished canonical basis for maximal cuts where the differential equations exhibit some additional symmetry. The rotation to this distinguished basis may require one to introduce additional square roots, and this was interpreted as the existence of a Galois symmetry relating different Feynman integrals. 
The main goal of this paper is to argue that the recently proposed novel concept of self-duality for maximal cuts introduced ref.~\cite{Pogel:2024sdi} is a consequence of the natural self-duality from twisted cohomology, but a specific choice of basis is enforced. Our main technical result is the proof that the intersection matrix of co-cycles computed for a canonical basis (and its dual) is a constant. Moreover, we show that to every canonical basis one can associate a representation of a Lie algebra, and the constant intersection matrix can then be interpreted as the change of basis that relates this representation and its dual. For maximal cuts, these representations turn out to be self-dual. This enforces constraints on the differential equations satisfied by the integrals, which in some cases implies the observation of ref.~\cite{Pogel:2024sdi}, but with a particular basis choice. Our formalism also allows us to pinpoint cases in which the observation of ref.~\cite{Pogel:2024sdi} may not hold. Our results show that self-duality, when interpreted in the right way, is independent of the choice of basis. We find that the choice of basis advocated in ref.~\cite{Pogel:2024sdi} is purely ad hoc, and the Galois symmetry observed in ref.~\cite{Pogel:2024sdi} is simply a consequence of this basis choice. 

Our paper is organised as follows: in section~\ref{sec:FIs} we briefly review Feynman integrals and the observation of ref.~\cite{Pogel:2024sdi}, and in section~\ref{sec:self-duality} we review some basics of twisted cohomology theories. In section~\ref{sec:eps-fac-DEQ} we present our main result, and we prove that the intersection matrix computed for canonical bases is constant. In section~\ref{sec:main:thm} we show how self-duality puts constraints on the differential equation, and in section~\ref{sec:galois} we comment on Galois symmetries. In section~\ref{sec:conclusion} we draw our conclusions.


\section{Feynman integrals and their cuts}
\label{sec:FIs}


\subsection{Feynman integrals and their differential equations}
The main objects of interest in this paper are Feynman integrals in dimensional regularisation in $D=d-2\eps$ dimensions:
\begin{align}
\label{eq:FI_def}
I_{\bs{\nu}}^D \left(\{p_i\cdot p_j\},\{m_i^2\} \right)= e^{L\gamma_E \varepsilon}\int \left(\prod_{j=1}^L \frac{d^D \ell_j}{i\pi^\frac{D}{2}}\right) \frac{1}{D_1^{\nu_1}\dots D_m^{\nu_m}}\, , 
\end{align}
where $d$ is an integer, $\bs{\nu} = (\nu_1,\ldots,\nu_m)$ is a vector of integers and $\gamma_E=-\Gamma'(1)$ is the Euler-Mascheroni constant. The momenta flowing through the $m$ propagators are linear combinations of the $E$ linearly independent external momenta $p_j$ and the $L$ loop momenta $\ell_j$, and we denote the squared mass of the $i^{\textrm{th}}$ propagator by $m_i^2$.

Feynman integrals that only differ by the value of $\bs{\nu}$ define a family. It is well known that not all members of a given family are independent, but there are linear relations among them. These so-called \emph{integration-by-parts} (IBP) relations~\cite{Tkachov:1981wb,Chetyrkin:1981qh} are succinctly captured by the identity
\beq\label{eq:IBP_tot_diff}
\int \rd^D \ell_i\,\frac{\partial}{\partial \ell_i^\mu}\left(\frac{v^{\mu}}{D_1^{\nu_1}\dots D_m^{\nu_m}} \right)=0\,,
\eeq
where $v^\mu$ can be either an internal or external momentum. Differentiation then produces propagators with shifted exponents (plus numerator factors that can again be expressed in terms of inverse propagators), leading to a set of linear relations among integrals from the same family, but with shifted values of the exponents $\nu_i$. The coefficients appearing in these relations are rational functions of the kinematic scales and the dimensional regulator $\eps$. One can solve the IBP relations and express all members of a family in terms of some basis integrals, often called \emph{master integrals} in the literature. The number of master integrals is always finite~\cite{Smirnov:2010hn,Bitoun:2017nre}.

There are a variety of techniques to evaluate the master integrals. One of the most prominent ones is the differential equations technique~\cite{Kotikov:1990kg,Kotikov:1991hm,Kotikov:1991pm,Gehrmann:1999as,Henn:2013pwa}. Differentiation with respect to an external momentum or scale leads again to a linear combination of integrals with shifted exponents, which can themselves be written in terms of master integrals. One obtains in this way a linear system of first-order differential equations satisfied by the master integrals:
\beq\label{eq:DEQ_prototype}
\rd \bI(\bx,\eps) = \bOmega(\bx,\eps)\bI(\bx,\eps)\,,
\eeq 
where $\bx= (x_1,\ldots,x_r)$ denotes the vector of dimensionless kinematic variables on which the integrals depend,\footnote{By dimensional analysis, the integral can only depend non trivially on dimensionless ratios of scales.} $\rd = \sum_{i=1}^r\rd x_i\,\partial_{x_i}$ is the total differential with respect to $\bx$ and $\bI(\bx,\eps) = \big(I_1(\bx,\eps),\ldots,I_N(\bx,\eps)\big)^T$ is the vector of master integrals. $\bOmega(\bx,\eps)$ is a matrix of rational one-forms in $\bx$ and it is rational in $\eps$, i.e., $\bOmega(\bx,\eps)$ can be written in the form
\beq
\bOmega(\bx,\eps) = \sum_{i=1}^r\rd x_i\,\bOmega_i(\bx,\eps)\,,
\eeq
where the $\bOmega_i(\bx,\eps)$ are matrices of rational functions in $\bx$ and $\epsilon$. 
There is a natural partial order on the master integrals (given by the propagators, or equivalently, which entries in $\bnu$ are positive), and if the master integrals are ordered in a way that respects this partial ordering, then the matrix $\bOmega(\bx,\eps)$ is block upper-triangular. Note that, since $\rd^2=0$, we must have
\beq\label{eq:flatness}
\rd\bOmega(\bx,\eps) = \bOmega(\bx,\eps) \wedge \bOmega(\bx,\eps) \,.
\eeq


There is a considerable freedom and arbitrariness in how to fix a basis of master integrals. If we change to a new basis $\bI'(\bx,\eps)$ via $\bI(\bx,\eps) = \bM(\bx,\eps)\bI'(\bx,\eps)$, then the new vector of master integrals satisfies the differential equation
\beq
\rd \bI'(\bx,\eps) = \bOmega'(\bx,\eps)\bI'(\bx,\eps)\,,
\eeq 
with
\beq\label{eq:gauge_transformation}
\bOmega'(\bx,\eps) =  \bM(\bx,\eps)^{-1}\Big(\bOmega(\bx,\eps) \bM(\bx,\eps) - \rd \bM(\bx,\eps)\Big)\,.
\eeq
A judicious choice of basis can have a strong impact on how difficult it is to solve the system of differential equations. 
In applications one is not interested in the exact dependence on $\eps$, but only in the first few terms in the Laurent expansion around $\eps=0$. In ref.~\cite{Henn:2013pwa} it was observed that there is always a distinguished basis $\bJ(\bx,\eps) =  \bS(\bx,\eps)^{-1}\bI(\bx,\eps)$ such that
\beq\label{eq:DEQ_can_basis}
\rd\bJ(\bx,\eps) = \eps\bA(\bx)\bJ(\bx,\eps)\,.
\eeq
The matrix $\bA(\bx)$ is related to $\bOmega(\bx,\eps)$ via eq.~\eqref{eq:gauge_transformation}:
\beq\label{eq:gauge_to_canonical}
\eps\bA(\bx) = \bS(\bx,\eps)^{-1}\Big(\bOmega(\bx,\eps)\bS(\bx,\eps)-\rd\bS(\bx,\eps)\Big)\,.
\eeq
The main advantage of such a so-called \emph{$\eps$-factorised} or \emph{canonical} basis lies in the fact that the solution of eq.~\eqref{eq:DEQ_can_basis} can easily be written down as a path-ordered exponential
\beq
\bJ(\bx,\eps) = \bU_{\gamma}(\bx,\eps)\bJ_0(\eps)\,,
\eeq
where $\bJ_0(\eps)$ denotes the value of $\bJ(\bx,\eps)$ at some point $\bx=\bx_0$ and $\gamma$ is a path from $\bx_0$ to a generic point $\bx$, and we defined the path-ordered exponential 
\beq\label{eq:Pexp_def}
\bU_{\gamma}(\bx,\eps) = \mathbb{P}\exp\left[\eps\int_{\gamma}\bA(\bx)\right]\,.
\eeq
Note that in this basis eq.~\eqref{eq:flatness} takes the form 
\beq
\label{eq:flatness_canonical}
\rd\bA(\bx) = \bA(\bx)\wedge\bA(\bx) = 0\,,
\eeq
and expresses the fact that the path-ordered exponential does not depend on the details of the path $\gamma$, but only on the end-points and the homotopy-class of the path. The path-ordered exponential can easily be expanded in $\eps$ up to a certain order, and the coefficients of the expansion will involve iterated integrals~\cite{ChenSymbol} built out of the one-forms $\omega_i$, with $\bA(\bx) = \sum_i\bA_i\omega_i$,
\beq\label{eq:Pexp_def_exp}
\bU_{\gamma}(\bx,\eps) = \mathds{1} + \sum_{k=1}^\infty \eps^k\sum_{1\le i_1,\ldots,i_k\le p}\bA_{i_1}\cdots\bA_{i_k}  I_{\gamma}(\omega_{i_1},\ldots,\omega_{i_k}) \,,
\eeq
where we defined the iterated integral:
\beq\label{eq:iterated_int_def}
 I_{\gamma}(\omega_{i_1},\ldots,\omega_{i_k}) = \int_{\gamma} \omega_{i_1}\cdots \omega_{i_k} = \int_{0\le \xi_k\le \cdots \le \xi_1\le1}\rd \xi_1 \,f_{i_1}(\xi_1)\,\rd \xi_2 \,f_{i_2}(\xi_2)\cdots \rd \xi_k\,f_{i_k}(\xi_k)\,.
 \eeq
 
 Besides the fact that $\eps$ factorises, there is another important difference between the matrices $\bOmega(\bx,\eps)$ and $\bA(\bx)$. While the entries of $\bOmega(\bx,\eps)$ are rational one-forms in $\bx$, the entries of $\bA(\bx)$ may not be rational. Indeed, the matrix $\bS(\bx,\eps)$ that expresses the change of basis from the original basis to the canonical one will in general involve the maximal cuts of the Feynman integrals in integer dimensions, and the latter often contain algebraic functions of the external kinematics and/or the periods and quasi-periods of some algebraic variety. This non-rational dependence then typically feeds into the matrix $\bA(\bx)$ via eq.~\eqref{eq:gauge_to_canonical}. There are examples where $\bA(\bx)$ involves modular forms~\cite{Adams:2018yfj,Broedel:2018rwm,Adams:2018bsn,Dlapa:2022wdu,Muller:2022gec,Pogel:2022ken,Klemm:2024wtd}, the coefficients of the Kronecker-Eisenstein series~\cite{Bogner:2019lfa,Muller:2022gec} or periods of Calabi-Yau varieties and integrals thereof~\cite{Pogel:2022ken,Pogel:2022vat,Pogel:2022yat,Gorges:2023zgv,Klemm:2024wtd}.
 
We stress that the existence of a canonical form for the differential equation is still conjectural, but by now there is an overwhelming body of evidence supporting the conjecture.
In particular, in the case where the differential forms $\omega_i$ are $\dlog$-forms, there is a solid understanding of how to find the canonical basis, see, e.g., refs.~\cite{Henn:2014qga,Lee:2014ioa,Gituliar:2017vzm,Meyer:2017joq,Lee:2020zfb,Henn:2020lye,Dlapa:2020cwj}. However, it is known that not all Feynman integrals admit a differential equation in canonical $\dlog$-form. Nevertheless, there are also many examples of Feynman integrals associated to non-trivial geometries of elliptic or Calabi-Yau type that admit a differential equation in canonical form, cf., e.g., refs.~\cite{Adams:2018yfj,Broedel:2018rwm,Adams:2018bsn,Bogner:2019lfa,Dlapa:2022wdu,Muller:2022gec,Pogel:2022ken,Pogel:2022vat,Pogel:2022yat,Gorges:2023zgv,Klemm:2024wtd,Ahmed:2024tsg}. The conjectured existence of a canonical form with the desired properties is thus on solid ground, and it is expected that the conjecture always holds (at least for the classes of associated geometries that have so far been encountered for Feynman integrals).

\subsection{Maximal cuts and self-duality}
\label{sec:maximal-cuts}
In applications, one is often also interested in cuts of Feynman integrals. The latter are defined, loosely speaking, by putting some of the propagators on their mass shell, which can be mathematically implemented by a residue prescription (cf., e.g., ref.~\cite{Britto:2024mna} for a recent review). The operation of taking residues preserves the linear relations and the differential equations, a feature known as \emph{reverse-unitarity} in the physics literature~\cite{Anastasiou:2002yz,Anastasiou:2003yy}. It follows that cuts of master integrals satisfy the same differential equations as their uncut analogues, but we need to put to zero all master integrals where we take a residue at a propagator raised to a negative power. Of special interest in this paper will be so-called \emph{maximal cuts}, which correspond to integrals where all propagators are put on their mass shell. Equivalently, a maximal cut of a Feynman integral is obtained by evaluating the integrand on a contour that encircles all the propagator poles. The contour is not specified otherwise, and the number of independent maximal cut contours is always equal to the number of master integrals with the given set of propagators. The differential equations for the maximal cuts are also particularly simple: they are given by the matrices appearing on the blocks on the diagonal of $\bOmega(\bx,\eps)$. 

We now focus on maximal cuts.
Let us consider a family of Feynman integrals as in eq.~\eqref{eq:FI_def}, and we assume that there are $N$ master integrals that have the maximal set of propagators. Then there are also precisely $N$ independent maximal cut contours, and we can form an $N\times N$ matrix $\bF(\bx,\eps)$ of maximal cuts associated to this family. By reverse-unitarity, this matrix fulfills a differential equation of the type
\beq
\rd\bF(\bx,\eps) = \bOmega^{\textrm{m.c.}}(\bx,\eps)\bF(\bx,\eps)\,,
\eeq
where $ \bOmega^{\textrm{m.c.}}(\bx,\eps)$ is an $N\times N$ matrix of rational one-forms. Based on the discussion of the previous section, it is expected that we can change to a canonical basis, i.e., we expect that there is a matrix $\bS^{\textrm{m.c.}}(\bx,\eps)$ such that if we let 
\beq\label{eq:FSG}
\bF(\bx,\eps) = \bS^{\textrm{m.c.}}(\bx,\eps)\bG(\bx,\eps)\,,
\eeq
then 
\beq
\rd\bG(\bx,\eps)= \eps\bA^{\textrm{m.c.}}(\bx)\bG(\bx,\eps)\,.
\eeq

In ref.~\cite{Pogel:2024sdi} it was observed that maximal cuts have the property of \emph{self-duality}, which according to ref.~\cite{Pogel:2024sdi} can be stated as follows: there is a constant matrix $\bM$ such that $\widetilde{\bA}^{\textrm{m.c.}}(\bx)=\bM^{-1}\bA^{\textrm{m.c.}}(\bx)\bM$ is persymmetric, i.e., it is symmetric with respect to the antidiagonal. This is equivalent to 
\beq\label{eq:weinzierl_sd}
\widetilde{\bA}^{\textrm{m.c.}}(\bx)= \bK_N\widetilde{\bA}^{\textrm{m.c.}}(\bx)^T\bK_N^{-1}\,,
\eeq
with
\beq\label{eq:K_matrix}
\bK_N = \bK_N^T=\bK_N^{-1}= \left(\begin{smallmatrix}
0 & 0 & \ldots & 0& 1\\
0 & 0 &  & 1& 0\\
\vdots &  & {.^{.^{.^{.^.}}}} & & \vdots\\
\phantom{.}&&&&\\
0 & 1 & \ldots & 0& 0\\
1 & 0 &  & 0& 0\\
\end{smallmatrix}\right)\,.
\eeq
Self-duality in the form of eq.~\eqref{eq:weinzierl_sd} can be interpreted as some sort of hidden symmetry of maximal cuts.
It was first observed in the context of multi-loop equal-mass banana integrals in refs.~\cite{Pogel:2022ken,Pogel:2022vat,Pogel:2022yat}, which are related to one-parameter families of Calabi-Yau varieties, and self-duality is expected to be a property of the Gauss-Manin connection for $\eps=0$ for such families. In ref.~\cite{Pogel:2024sdi}, self-duality for maximal cuts was observed to hold more generally, based on the explicit analysis of differential equations for some Feynman integrals and their maximal cuts. The goal of this paper is to connect the self-duality from ref.~\cite{Pogel:2024sdi} to a more general notion of self-duality that arises naturally from twisted cohomology, which is known to provide the appropriate mathematical framework to study Feynman integrals in dimensional-regularisation~\cite{Mastrolia:2018uzb}.


\section{Twisted cohomology and self-duality}
\label{sec:self-duality}

\subsection{Twisted cohomology}
\label{sec:twisted}

Our primary goal is to understand self-duality for maximal cuts, whose natural framework is twisted cohomology. 
More precisely, we want to consider integrals of the form
\begin{align}
\label{gen_twisted_integral}
    \int_{\Gamma} \Phi \varphi\, , 
\end{align}
where $\varphi$ is a rational $n$-form on $X=\mathbb{C}^n-\Sigma$, with $\Sigma$ a union of hypersurfaces and $\Phi$ is a multi-valued function of the form
\begin{align}
\label{twist32}
     \Phi= \prod_{i=0}^{r} L_i(\bs{z})^{\alpha_i} \,.
\end{align}
The $L_i(\bs{z})$ are polynomials in the integration variables $\bs{z}\in X$ and $\Sigma$ contains the loci $L_i(\bs{z})=0$. In the most general case $\alpha_i\in\mathbb{C}$, but we assume the exponents $\alpha_i$ to be generic here:
\begin{align}
\label{restrict}
\alpha_i\notin \mathbb{Z}\text{~~~~~and~~~~~} \sum_{i=0}^{r} \alpha_i \notin \mathbb{Z}\, .
\end{align}

Maximal cuts of Feynman integrals naturally match the class of integrals defined by eq.~\eqref{gen_twisted_integral} with the genericity assumption in eq.~(\ref{restrict}). This can for example easily be seen in the Baikov representation of a Feynman integral~\cite{BAIKOV1997347} (see also ref.~\cite{LEE2010474}), where the twist is identified with the Baikov polynomial, and the multi-valuedness comes from the fact that the exponent of the Baikov polynomial contains the dimensional regulator $\eps$, which is not an integer. For maximal cuts, all possible poles of $\varphi$ are already regulated by the twist, so $\varphi$ introduces no further poles. In the following, we provide a very brief review of twisted cohomology, focusing on the aspects that are needed to understand self-duality for maximal cuts. For a more in-depth review of twisted (co-)homology and its applications we refer to the literature \cite{yoshida_hypergeometric_1997,aomoto_theory_2011,Mizera:2017rqa,Mizera:2019gea,Mizera:2019ose,matsumoto_relative_2019-1,Caron-Huot:2021iev,Caron-Huot:2021xqj,Giroux:2022wav,Crisanti:2024onv,Gasparotto:2023roh,Fontana:2023amt,Bhardwaj:2023vvm,De:2023xue,Britto:2021prf,Duhr:2023bku,Brunello:2023rpq,Duhr:2024rxe}.

The twist defines a connection,
\begin{align}
\label{connect}
\nabla = \rd +\omega \wedge \cdot \text{~~~~with~~~~} \omega = \frac{\rd \Phi}{\Phi} =\rd\!\log \Phi\, \, , 
\end{align}
such that total derivatives of $\Phi\varphi$ are equivalent to computing the covariant derivative of $\varphi$:
\beq
\rd\big(\Phi\varphi\big) = \Phi\,\nabla\varphi\,.
\eeq
Total derivatives of $\Phi\varphi$ integrate to boundary terms. We thus want to identify classes of integrands modulo total derivatives, which is equivalent to identifying rational $n$-forms $\varphi$ modulo covariant derivatives, $\varphi \sim \varphi +\nabla \tilde{\varphi}$.
This naturally leads one to consider
the  \textit{twisted cohomology group} 
\beq
\label{twistedcoh}
    H_{\text{dR}}^n (X, \nabla)=C^n(X,\nabla)/B^n(X,\nabla)\,, 
    \eeq
    with
    \beq\bsp
    \label{CBeq}
    C^n(X,\nabla)&\,=\{ n-\text{forms }\varphi \text{ on } X\, :\, \nabla \varphi = 0 \}\,,\\
    B^n(X,\nabla)&\,= \{ n-\text{forms }\nabla\tilde{\varphi}\, :\, \tilde{\varphi} \text{ a $(n-1)$-form}\} \, . 
\esp\eeq
The elements of $H_{\text{dR}}^{n}(X,\nabla)$ are called \textit{twisted co-cycles}. The relevant integration contours are \textit{twisted cycles} (also called \emph{loaded cycles}) from 
\begin{align}
\label{homgroup}
    C_n(X,\check{\mathcal{L}}) =  \{n-\text{cycles }\gamma \otimes \Phi|_{\gamma}\,  :\, \partial (\gamma\otimes \Phi|_{\gamma} )=0\} \, . 
\end{align}
Here $\gamma \otimes \Phi|_\gamma$ denotes an $n$-cycle $\gamma$ on $X$ together with a local choice of branch for $\Phi$, specified by the locally constant sheaf $\check{\mathcal{L}}$. 
The operation $\partial (\gamma\otimes\Phi|_{\gamma})$ restricts the contour with its branch of $\Phi$ to the boundary. If we only consider integration cycles modulo boundaries, we are led to consider the
\textit{twisted homology group}, defined by
\begin{align}
\label{twistedho}
H_n(X, \check{\mathcal{L}}) =C_n(X,\check{\mathcal{L}})/ B_n(X,\check{\mathcal{L}}) \, , 
\end{align}
where $B_n(X,\check{\mathcal{L}})= \{n-\text{cycles }\partial \left(\gamma\otimes \Phi|_{\gamma}\right)\}$ denotes the space of all twisted boundaries. All cycles relevant in this paper are regularised chambers between regulated boundaries. The details of the regularisation are explained in refs.~\cite{yoshida_hypergeometric_1997,kita_intersection_1994-2,aomoto_theory_2011, Bhardwaj:2023vvm, Duhr:2023bku}.
We note that for the purposes of evaluating integrals, the details of the regularisation can usually be ignored, and the twisted cycles can be treated as non-regularised chambers between boundaries, see, e.g., ref.~\cite{yoshida_hypergeometric_1997}.

Having identified the relevant integrands and contours as the twisted co-cycles and cycles respectively, we can pair them to obtain integrals such as those in eq.~\eqref{gen_twisted_integral}: 
\begin{align}
 \langle \gamma |\varphi]= \int_{\gamma} \Phi \varphi\, . 
\end{align}
If we choose bases $\{\gamma_i\}$ of $H_n(X,\check{\mathcal{L}})$ and $\{\varphi_j\}$ of $H_{\text{dR}}^n(X,\nabla)$, we can pair the basis elements to obtain the period matrix $\bs{P}$ with entries
\begin{align}
 P_{ij}= \langle\varphi_i| \gamma_j]\  =\int_{\gamma_j}\Phi\varphi_i\, . 
\end{align}
The twisted (co-)homology groups are vector spaces, and so it is natural to consider their duals. 
Dual cocycles form equivalence classes modulo covariant derivatives with respect to the connection $\check{\nabla}=\rd -\omega\wedge \cdot\, $. We will give a concrete description of the (dual) (co-)cycles in section~\ref{sec:twisted-self-duality}. We can fix bases $\{\check{\gamma}_i\}$ and $\{\check{\varphi}_i\}$ of the dual twisted (co-)homology groups. This allows us to define the dual period matrix $\bs{\check{P}}$ by pairing dual cycles with dual co-cycles:
\begin{align}
    \check{P}_{ij} = [\check{\gamma}_j|\check{\varphi}_i\rangle = \int_{\check{\gamma}_j} \Phi^{-1} \check{\varphi}_i\, .  
\end{align}
In addition, we can define two more pairings, called \emph{intersection pairings}, between the twisted (co-)cycles and their duals. They give rise to intersection matrices between twisted (co-)cycles:
\beq\bsp
    C_{ij} &\,= \frac{1}{(2\pi i)^n} \langle \varphi_i |\check{\varphi}_j\rangle = \frac{1}{(2\pi i)^n}\int_X \varphi_i\wedge \check\varphi_{j}\, ,\\
    H_{ij}&\,=[\check{\gamma}_j|\gamma_i]\, ,
\esp\eeq
where $[\check{\gamma}_i|\gamma_j]$ counts the (topological) intersections of the two contours, taking into account their orientations as well as the branch choices for $\Phi$ and $\Phi^{-1}$ loaded onto them. 

The matrices $\bP$, $\bs{\check{P}}$, $\bC$ and $\bH$ will play an important role in the remainder of this paper. We therefore summarise some of their properties. First, all four pairings are non-degenerate, and so the matrices $\bP$, $\bs{\check{P}}$, $\bC$ and $\bH$ have full rank (at least for generic values of $\bx$ and $\balpha = (\alpha_1,\ldots,\alpha_r)$). 
Second, in applications all quantities may depend on some parameters $\bx$ (the external kinematic data in the case of Feynman integrals) as well as the exponents $\balpha$ appearing in the twist. The entries of the twisted period matrix and its dual are generally transcendental functions of $\bx$, while the matrices $\bC$ only depend rationally on $\bx$, and $\bH$ is constant in $\bx$. The $\bx$-dependence of the matrices is governed by the following differential equations (cf., e.g., refs.~\cite{Chestnov:2022alh,Chestnov:2022okt,Munch:2022ouq})
\beq\bsp\label{eq:Gauss-Manin}
\rd \bP(\bx,{\balpha}) &\,= \bs\Omega(\bx,\balpha)\bP(\bx,\balpha)\,,\\
\rd \bs{\check{P}}(\bx,{\balpha}) &\,= \bs{\check{\Omega}}(\bx,\balpha)\bs{\check{P}}(\bx,\balpha)\,,\\
\rd\bC(\bx,\balpha) &\,= \bs\Omega(\bx,\balpha)\bC(\bx,\balpha) + \bC(\bx,\balpha)\bs{\check{\Omega}}(\bx,\balpha)^T\,,\\
\rd\bH(\balpha) &\,=0\,,
\esp\eeq
where the matrices $\bs\Omega(\bx,\balpha)$ and $\bs{\check{\Omega}}(\bx,\balpha)$ are matrices of rational one-forms of $\bx$ and they depend rationally on $\balpha$. Finally, we emphasise that the four matrices $\bP$, $\bs{\check{P}}$, $\bC$ and $\bH$ are not independent, but they are related by the \textit{twisted Riemann bilinear relations} (TRBRs)~\cite{Cho_Matsumoto_1995}
\beq\bsp
\label{generalriemann}
 {(2\pi i)^{-n}}\bs{P}(\bx,\balpha) \left(\bs{H}(\balpha)^{-1}\right)^T \bs{\check{P}}(\bx,\balpha)^T&\, =    \bs{C}(\bx,\balpha) \, . 
\esp\eeq

For non-maximal cuts or uncut Feynman integrals, the framework of twisted cohomology is not applicable in the way stated above, but we need to use  \textit{relative twisted cohomology},\footnote{In many applications one can also instead  deform the exponents of the unregulated poles and recover the same results as in the relative cohomology from a certain limit \cite{Brunello:2023rpq,Duhr:2024rxe}.} because the rational form $\varphi$ then collects  poles coming from the propagators, as well as possible numerator factors that are not regulated by the twist. 
The same conclusion can be reached by working directly in momentum space~\cite{Caron-Huot:2021xqj}. We will not give a detailed review of relative twisted cohomology, but refer to the literature ~\cite{matsumoto_relative_2019-1,Caron-Huot:2021xqj,Caron-Huot:2021iev}. We just note that most of the concepts and objects for twisted cohomology introduced above appear in a very similar way in relative twisted cohomology and the main difference lies in the way we choose and regulate the bases of twisted cycles and cocyles \cite{Caron-Huot:2021xqj,Caron-Huot:2021iev}.


\subsection{Self-duality from twisted cohomology}
\label{sec:twisted-self-duality}
We have already mentioned that non-maximally cut Feynman integrals require the framework of relative twisted cohomology, while maximal cuts can be treated within the non-relative framework reviewed in the previous subsection. 
We now introduce a property of non-relative cohomology that leads to a direct definition of a notion of self-duality. 

If condition~\eqref{restrict} holds, the dual (co-)homology groups are given by
\beq\bsp
H_{\text{dR},c}^n \left(X, \check{\nabla}\right) &= \{ \text{compactly supported } n-\text{forms } \check{\varphi} \, :\, \check{\nabla} \check{\varphi} =0\}/ \{\text{exact forms}\}\,,\\
H_n^{\text{lf}}(X,\mathcal{L}) &=  \{\check{\gamma} \otimes \Phi^{-1}|_{\check\gamma} \text{ locally finite} \, :\,  \partial \check{\gamma}  =0\} / \{ \text{boundaries }\partial\tilde{\gamma }\} \, .
\esp\eeq
For a given $n$-form $\varphi$, one can compute a compactly supported version in the same cohomology class with the methods described  in refs.~\cite{Mizera:2017rqa,Mastrolia:2018uzb}. We denote this compactly supported version by $\left[\varphi\right]_c$. We can then choose as a basis of the dual twisted cohomology group the compactly supported version of the basis of the twisted cohomlogy group:
\beq\label{eq:check_to_c}
\check{\varphi}_i = [\varphi_i]_c\,.
\eeq
Similarly, a basis cycle can be chosen to be the regularised version of the dual basis cycle~\cite{Cho_Matsumoto_1995}:
\beq\label{eq:check_to_h}
\gamma_i=[\check{\gamma}_i ]_\text{reg}\,.
\eeq
When integrating, we obtain
\beq
\check{P}_{ij}(\bx,\balpha)=\int_{\check{\gamma}_j }\Phi^{-1}\check{\varphi}_i=\int_{\check{\gamma}_j }\Phi^{-1}[\varphi_i]_c=\int_{[\check{\gamma}_j]_\text{reg} }\Phi^{-1}\varphi_i=\int_{\gamma_j}\Phi^{-1}\varphi_i= P_{ij}(\bx,-\balpha)\,.
\eeq
Thus, if condition~\eqref{restrict} holds, the difference between the period matrix and its dual only lies in the choice of the twist $\Phi^{-1}$ instead of $\Phi$. We obtain in this way a self-duality property, which allows us to identify the dual twisted period matrix $\bs{\check{P}}$ with  twisted period matrix $\bP$, up to changing the signs of the exponents $\alpha_i$:
\beq\label{eq:twisted-duality-general}
\bs{\check{P}}(\bx,\balpha) = \bP(\bx,-\balpha)\,,\qquad \textrm{if condition~\eqref{restrict} holds}.
\eeq
This self-duality has several consequences. First, we see that the matrices appearing in the differential equations~\eqref{eq:Gauss-Manin} are related:
\beq\label{eq:self-dual-Omega}
\bs{\check{\Omega}}(\bx,\balpha) = \bOmega(\bx,-\balpha)\,.
\eeq
Second, we see that the TRBRs in eq.~\eqref{generalriemann} now reduce to a quadratic relation for the period matrix (albeit one of them evaluated with $\balpha\to -\balpha$)~\cite{Duhr:2024rxe}:
\beq\bsp
\label{self-dual-riemann}
{(2\pi i)^{-n}} \bs{P}(\bx,\balpha) \left(\bs{H}(\balpha)^{-1}\right)^T \bs{P}(\bx,-\balpha)^T&\, =     \bs{C}(\bx,\balpha) \,.
\esp\eeq

Let us conclude this section by discussing how self-duality is reflected at the level of maximal cuts.
In that case, the twist generally is of the form such that all exponents have the form $\alpha_j=\pm\eps-\frac{m_j}{2}$, with $m_j\in\{0,1\}$. Introducing the set $J:=\{ j:m_j=1\}$, we may choose the basis of dual co-cycles to be
\beq
\check{\varphi}_{i}=\left[\varphi_i \prod_{j\in J}L_j(z)^{-1}\right]_c\,. 
\eeq
With this choice eqs.~\eqref{eq:twisted-duality-general} and~\eqref{eq:self-dual-Omega} translate into~\cite{Duhr:2024rxe}
\beq\bsp
\check{\bP}(\bx,\varepsilon)&=\bP(\bx,-\varepsilon)\,,\\
\check{\bOmega}(\bx,\eps)&=\bOmega(\bx,-\eps)\label{dualomegaepsilon}\,,
\esp\eeq
and the TRBRs in eq.~\eqref{self-dual-riemann} reduce to quadratic relations among maximal cuts,
\beq\bsp
\label{FI:self-dual-riemann}
{(2\pi i)^{-n}} \bs{P}(\bx,\eps) \left(\bs{H}(\eps)^{-1}\right)^T \bs{P}(\bx,-\eps)^T&\, =     \bs{C}(\bx,\eps) \,.
\esp\eeq

\section{$\eps$-factorised differential equations and their properties}
\label{sec:eps-fac-DEQ}

\subsection{$\eps$-factorised and canonical differential equations}
\label{sec:eps-fac-DEQ_def}
The observation of self-duality from ref.~\cite{Pogel:2024sdi} asserts that the diagonal blocks of the matrix defining the differential equations in canonical form have the symmetry in eq.~\eqref{eq:weinzierl_sd}. Our goal is to see if we can connect this observation to the self-duality that naturally arises from twisted cohomology, thereby getting some hint on the validity of the observation, and ideally even proving it. 

At this point, however, we need to address an issue: Equation~\eqref{eq:weinzierl_sd} expresses a property of the differential equation satisfied by the maximal cuts when rotated to the canonical basis $\bG(\bx,\eps)$. The existence of such a basis, however, is still conjectural, and so it is hard to know if such a basis even exists, and if it does, what its properties are in the general case. We therefore take a different approach: we will define a class of differential equations that captures (some of) the properties that we know to hold for all known examples of differential equations in canonical form, and that is sufficient to provide rigorous proofs. In this way we can guarantee that our proofs are solid and do not rely on the conjectured existence of a system of differential equations in canonical form. 

Let us describe the class of differential equations that we want to consider. Let $\KK$ be an algebraic number field, i.e., a field obtained by adjoining to $\mathbb{Q}$ a finite number of solutions of an algebraic equation. We consider a differential equation of the form
\beq\label{eq:eps-factorised}
\rd \bG(\bx,\eps) = \eps\bA(\bx)\bG(\bx,\eps)\,,
\eeq
where $\bG(\bx,\eps)$ is an $N\times N$ matrix, and the entries of the matrix $\bA(\bx)$ are one-forms. We call a differential equation as in eq.~\eqref{eq:eps-factorised} \emph{$\eps$-factorised}.\footnote{In the literature the names \emph{$\eps$-factorised systems} and \emph{systems in canonical form} are often used interchangeably.  Since the factorisation of $\eps$ is only one of the properties one expects to hold, we prefer to keep the two concepts distinct.} Note that, due to the flatness condition~\eqref{eq:flatness_canonical}, all differential forms must be closed. We assume that around every point we can find a local coordinate chart $\bx$ such that the entries of $\bA(\bx)$ take the form
\beq\label{eq:A_ij_form}
A_{ij}(\bx) = \sum_{k=1}^r\rd x_k\,f_{ijk}(\bx)\,,
\eeq 
where the $f_{ijk}(\bx)$ are chosen from a $\KK$-algebra $\cA$ of functions. We assume that $\cA$ is \emph{differentially closed}, by which we mean that for all $f\in\cA$, we have $\partial_{x_i}f\in\cA$, for all $1\le i\le r$ (if $\cA$ is not differentially closed, we can enlarge it by adding the necessary derivatives). We also assume that the field of constants of $\cA$ is $\KK$,
\beq
\Const(\cA) = \big\{f\in \cA:\partial_{x_i}f = 0\,, \textrm{ for all }1\le i\le r\big\} = \KK\,.
\eeq
Here are some examples of algebras $\cA$ that have appeared for Feynman integrals:
\begin{enumerate}
\item If a canonical basis can be found by an algebraic transformation, i.e., the transformation matrix $\bS(\bx,\eps)$ depends algebraically on $\bx$, then $\bA(\bx)$ is a matrix of algebraic one-forms. The algebra $\cA$ is then typically a ring of rational functions with a prescribed set of poles, to which one adjoins one or more square roots of polynomials. $\cA$ is obviously differentially closed (because the derivative of an algebraic function is algebraic).
\item In the case where $r=1$ the functions $f_{ijk}$ in eq.~\eqref{eq:A_ij_form} may be (meromorphic) modular forms. The algebra $\cA$ then contains the algebra of (meromorphic) modular forms. The latter, however, is not closed, because derivatives of modular forms are not modular. The algebra of quasi-modular forms contains the algebra of modular forms and is closed under differentiation~\cite{Zagier2008}. We therefore choose $\cA$ to be the algebra of (meromorphic) quasi-modular forms. See also example~\ref{ex2f1} below.
\end{enumerate}
Other examples involve functions on the moduli space of elliptic curves with marked points~\cite{Bogner:2019lfa,Giroux:2024yxu} and rings of algebraic functions extended by periods and quasi-periods of one-parameter families of Calabi-Yau varieties~\cite{Pogel:2022ken,Pogel:2022vat,Pogel:2022yat}.

We denote the $\KK$-vector space of differential forms of the form~\eqref{eq:A_ij_form} by $\Omega^1(\cA)$, and we denote the subspace of closed forms by
\beq
Z^1(\cA) = \big\{\omega\in \Omega^1(\cA):\rd \omega=0\big\}\,.
\eeq
%
$Z^1(\cA)$ may be infinite-dimensional. However, the entries of $\bA(\bx)$ generate a finite-dimensional subspace $\mathbb{A}$ of $Z^1(\cA)$. 
Let us fix a basis $\omega_1,\ldots,\omega_p$ of $\mathbb{A}$. The $\omega_i$ are sometimes called the \emph{letters} in the context of Feynman integrals. We can then expand $\bA(\bx)$ in that basis:
\beq
\bA(\bx) = \sum_{i=1}^p\bA_i\omega_i\,,
\eeq
where the $\bA_i$ are constant $N\times N$ matrices with entries in $\KK$.

Finally, we define $\cA_{\eps} := \cA\otimes_{\KK} \KK(\eps)$ and $\cF:=\Frac(\mathbb{C}\otimes_{\KK}\cA)$. $\cA_{\eps}$  contains the algebra $\cA$, but we now allow coefficients that are rational functions in $\eps$, and $\cF$ is the field of fractions of the algebra $\cA$, where we also enlarge the constants to be complex numbers.
We can now define our class of differential equations. 
\begin{definition}\label{def:canon}
An $\eps$-factorised differential equation is in C-form if $\mathbb{A}\cap \rd\cF = \{0\}$.
\end{definition} 
As far as we know, this definition applies to all known examples of systems in canonical form that have appeared for Feynman integrals, i.e., all known examples of systems in canonical form are also in C-form.
We expect that systems in canonical form may have additional properties, e.g., we expect them to have at most logarithmic singularities. The class of differential equations from Definition~\ref{def:canon}, however, is sufficient to prove our results and to connect them to self-duality. At the same time, since all known examples of systems in canonical form are also in C-form, the results from subsequent sections can immediately be transferred to the known systems in canonical form. 
We now discuss two examples that illustrate the concepts introduced in this section so far.

\begin{examplenew}[Gauss hypergeometric function]{1}\label{ex2f1}
As a first example, we consider the family of integrals 
\begin{align}\label{eq:T_def}
T(n_1,n_2,n_3)=\int_0^1 \rd z\, z^{-\tfrac{1}{2}+n_1+a \varepsilon}(1-z)^{-\tfrac{1}{2}+n_2+b \varepsilon}(1-xz)^{-\tfrac{1}{2}+n_3+c \varepsilon}\,,
\end{align}
with $a,b,c\in\mathbb{Q}\setminus\{0\}$ and $n_i \in \mathbb{Z}$. The integrals $T(n_1,n_2,n_3)$ can be expressed in terms of Gauss' ${_2}F_1$ function, up to a prefactor that is a ratio of $\Gamma$ functions. While the class of integrals defined by eq.~\eqref{eq:T_def} is not directly a family of Feynman integrals, it is the prime example of a family of integrals that can be described in terms of twisted cohomology, and so it is the simplest example to illustrate all concepts.

 A choice of master integrals of this family is given by $\bF(x,\varepsilon)=\big(T(0,0,0),T(1,0,0)\big)^T$, and it satisfies a differential equation of the type $\rd\bF(x,\eps) =\bOmega(x,\eps)\bF(x,\eps)$, where $\bOmega(x,\eps)$ is a matrix of rational one-form in $x$ and it is also rational in $\eps$. For the explicit expression of $\bOmega(x,\eps)$, we refer, e.g., to ref.~\cite{Broedel:2018rwm}. This differential equation is not in canonical form. It is possible to define a new basis of master integrals $\bG(x,\varepsilon)=\bS_T(x,\varepsilon)^{-1}\bF(x,\varepsilon)$ such that the differential equation for $\bG(x,\eps)$ is $\eps$-factorised:
 \begin{align}
\rd\bG(x,\varepsilon)=\varepsilon \bA_{T}(x)\bG(x,\varepsilon) \,.
\end{align}
The transformation matrix is explicitly given by
\cite{Broedel:2018rwm}
\begin{align}\label{eq:ST_def}
S_T(x,\varepsilon)=\left(
\begin{array}{cc}
 2 \K(x) & 0 \\
 \frac{4 \K(x)^2 (a (x+1) \varepsilon +b \varepsilon +c x \varepsilon +1)-4 \K(x) \E(x)+\varepsilon }{2 x \K(x) (2 \varepsilon  (a+b+c)+1)} & \frac{\varepsilon }{2 x \K(x) (2 \varepsilon  (a+b+c)+1)} \\
\end{array}
\right)\,,
\end{align}
where we defined the complete elliptic integrals of the first and second kind,
\beq
\K(x) = \int_0^1\frac{\rd }{\sqrt{(1-xz^2)(1-z^2)}}\textrm{~~~and~~~}\E(x) = \int_0^1\rd z\,\sqrt{\frac{1-xz^2}{1-z^2}}\,.
\eeq
In this case, $\mathcal{A}_{T}$ is the (differentially closed) $\mathbb{Q}$-algebra  
\begin{align}
\mathcal{A}_{T}=\mathbb{Q}\left[i\pi^{\pm1},\frac{1}{x},\frac{1}{1-x},x,\K(x),\E(x),\frac{1}{\K(x)}\right]\,.
\end{align}
In order to describe the matrix $\bA_T$, it is useful to change variables to
\beq
\tau = i\,\frac{\K(1-x)}{\K(x)}\,.
\eeq
The inverse of this relation is the modular $\lambda$ function, 
with
\beq
x=\lambda(\tau) = \frac{\theta_2(\tau)^4}{\theta_3(\tau)^4}\,,
\eeq
which is a Hauptmodul for the congruence subgroup $\Gamma(2)\subseteq \SL(2,\mathbb{Z})$, with
\beq
\Gamma(N) = \left\{\gamma \in\SL(2,\mathbb{Z}) : \gamma = \mathds{1}\!\!\!\!\mod N\right\}\,.
\eeq
Here $\theta_i(\tau)$ are the usual Jacobi $\theta$-functions. The matrix $\bA_T$ then takes the form
\beq\label{A2f1}
\bA_{T}(\tau)=\bA_{T,1}\,\omega_{T,1}+\bA_{T,2}\,\omega_{T,2}+\bA_{T,3}\,\omega_{T,3}\,,
\eeq with 
\beq\bsp\label{eq:omega_2F1}
\omega_{T,1}&\,=\frac{\rd \tau}{2 i\pi}\,,\\
 \omega_{T,2}&\,=\frac{\rd \tau}{2 i\pi}\,h_2(\tau)\,4\left[(a+b+(c-a)\lambda(\tau))\right]\,,\\
\omega_{T,3}&\,=\frac{\rd \tau}{2 i\pi}\,h_2(\tau)^2\,16\left[(a+b)^2+(a+c)^2 \lambda(\tau)^2-2(a^2+ab+ac-bc)\lambda(\tau)\right]\,,
\esp\eeq
and $h_2(\tau) = \tfrac{\pi^2}{4}\theta_3(\tau)^4 = \K(\lambda(\tau))^2$ is an Eisenstein series of weight two for $\Gamma(2)$. The matrices $\bA_{T,i}$ are given by
\begin{align}\label{singlea2f1}
\bA_{T,1}=\left(
\begin{smallmatrix}
 \phantom{-}1 &  \phantom{-}1 \\
 -1 & -1 \\
\end{smallmatrix}
\right)\,,
\qquad \bA_{T,2}=\left(
\begin{smallmatrix}
 1 & 0 \\
 0 & 1 \\
\end{smallmatrix}
\right)\,, \qquad
\bA_{T,3}=\left(
\begin{smallmatrix}
 0 & 0 \\
 1 & 0 \\
\end{smallmatrix}
\right)\,.
\end{align}
The vector space $\mathbb{A}_T$ is generated by the differential forms in eq.~\eqref{eq:omega_2F1},
\begin{align}
\mathbb{A}_{T}=\left\langle\omega_{T,1},\,\omega_{T,2},\,\omega_{T,3}\right\rangle_{\mathbb{Q}}\,.
\end{align}
If we pass to the variable $\tau$, then also $\cA_T$ can be described in terms of modular forms. However, since the algebra of modular forms is not differetially closed, we need to pass to the algebra of quasi-modular forms~\cite{Zagier2008}. More precisely, one sees that $\cA_T$ is the algebra of weakly\footnote{Weakly holomorphic modular forms are allowed to have poles at the cusps.} holomorphic quasi-modular forms for the congruence subgroup $\Gamma(2)$\footnote{Note that entries of $\bS_T(x,\eps)$ contain odd powers of $\K(x)$. The latter are not modular forms for $\Gamma(2)$, because $\Gamma(2)$ does not have modular forms of odd weights. $\K(x)$ defines a modular form of weight 1 for the subgroup $\Gamma(4)$.} with $i\pi^\pm$ adjoined. Finally, using the same strategy as in refs.~\cite{10.1093/imrn/rnaa060,10.2140/ant.2017.11.2113}, one may check that $\bbA_T\cap \rd\ccF_T=\{0\}$, with $\ccF_T = \Frac(\mathbb{C}\otimes_{\mathbb{Q}}\cA_T)$, and so this system is in C-form.
\end{examplenew}
\begin{examplenew}[Unequal-mass sunrise integral]{2}
\label{exSunrise}
    As a second example let us consider the two-loop sunrise integral with unequal internal masses in $D=2-2\varepsilon$ dimensions,
    \begin{equation}
        I_{\nu_1,\nu_2,\nu_3}=e^{2\gamma_\mathrm{E}\varepsilon}\int\frac{\rd^D \ell_1}{i\pi^{D/2}}\int\frac{\rd^D \ell_2}{i\pi^{D/2}}\frac{1}{D_1^{\nu_1}D_2^{\nu_2}D_3^{\nu_3}} \,,
    \end{equation}
    with the propagators
    \begin{equation}
        D_1=-\ell_1^2+m_1^2\,,\qquad D_2=-\ell_2^2+m_2^2\,,\qquad D_3=-(p-\ell_1-\ell_2)^2+m_3^2 \,.
    \end{equation}
    Here $p$ is the external momentum and $m_1,m_2,m_3$ are the internal masses.
    By dimensional analysis this integral depends on three variables, which can be taken as
    \begin{equation}
        y_0=\frac{p^2}{m_3^2}\,,\qquad y_1=\frac{m_1^2}{m_3^2}\,,\qquad y_2=\frac{m_2^2}{m_3^2}\,.
    \end{equation}
    By studying the maximal cut one can easily see that this integral is connected to an elliptic curve. It turns out to be convenient to change variables from $(y_0,y_1,y_2)$ to the variables $\bx=(\tau,z_1,z_2)$, where $\tau$ is the modular parameter of the elliptic curve and $z_1,z_2$ are two punctures on it.

    The unequal-mass sunrise integral has 7 master integrals, 4 of which in the top sector. A canonical form for the differential equations satisfied by all 7 integrals was achieved in ref.~\cite{Bogner:2019lfa}. The explicit change of basis to the canonical basis $\bJ_\mathrm{S}(\bx,\eps)$ was described in ref.~\cite{Bogner:2019lfa}, and in this basis the differential equations take the form,
    \begin{equation}
        \rd \bJ_\mathrm{S}(\bx,\eps)=\eps\bA_\mathrm{S}(\bx)\bJ_\mathrm{S}(\bx,\eps) \,,
    \end{equation}
    where $\bm{x}=(\tau,z_1,z_2)$. The entries of the differential equation matrix $\bA_\mathrm{S}(\bx)$ can be written in terms of the Eisenstein-Kronecker forms \cite{Broedel:2018iwv}
    \begin{equation}\label{eq:omega_EK}
        \omega^{(k)}_{\mathrm{EK}}(z,\tau)=(2\pi)^{2-k}\left[ g^{(k-1)}(z,\tau)\rd z+(k-1)g^{(k)}(z,\tau)\frac{\rd\tau}{2\pi i} \right]\,,
    \end{equation}
    evaluated at the points $z=z_1,z_2,z_3=1-z_1-z_2$ and with modular parameters $\tau$ or $2\tau$, as well as two modular forms $\eta_2(\tau)\rd\tau,\eta_4(\tau)\rd\tau$ for $\Gamma_0(2)$ and $\mathrm{SL}(2,\mathbb{Z})$, respectively (see ref.~\cite{Bogner:2019lfa} for their explicit definition), with
    \beq
    \Gamma_0(N) = \left\{\left(\begin{smallmatrix}a&b\\c&d\end{smallmatrix}\right)\in \SL(2,\mathbb{Z}) : c=0\!\!\!\!\mod N \right\}\,.
    \eeq

    The entries of the $\varepsilon$-factorised differential equation matrix $\bA_\mathrm{S}(\bx)$ span the space $\mathbb{A}_{\mathrm{S}}$ which is a subspace of the differentially closed algebra
    \begin{equation}
        \mathcal{A}_{\mathrm{S}}=\mathbb{Q}\big[i\pi^{\pm1}\big]\otimes_{\mathbb{Q}}\mathrm{QM}(\Gamma_0(2))\otimes_{\mathbb{Q}}\mathrm{G_{EK}}
    \end{equation}
    consisting of products of weakly holomorphic quasi-modular forms for $\Gamma_0(2)$ and (derivatives of) coefficients of the Eisenstein-Kronecker series \cite{Broedel_2018},
    \begin{equation}
        \mathrm{G_{EK}}=\left\langle \partial_z^kg^{(m)}(z_i,j\tau),\,\, m\geq0,\,k\geq 0,\,i=1,2,3,\,j=1,2\right\rangle_{\mathbb{Q}} \,.
    \end{equation}
    Iterated integrals built out of the differential forms~\eqref{eq:omega_EK} evaluate to elliptic polylogarithms and iterated integrals of modular forms. One can show that the latter are linearly independent over $\ccF_{\mathrm{S}} = \Frac(\mathbb{C}\otimes_{\mathbb{Q}}\mathcal{A}_{\mathrm{S}})$ (cf., e.g.,~refs.~\cite{10.1093/imrn/rnaa060,10.2140/ant.2017.11.2113,Broedel_2018,zerbini_enriquez}), from which it follows that $\bbA_{\mathrm{S}}\cap\rd\ccF_{\mathrm{S}}=\{0\}$, and so the system is in C-form (see also the next subsection).
\end{examplenew}

\subsection{Linear independence of iterated integrals}
\label{sec:line_dep}

In the next subsection, we want to prove our main result, namely, that the cohomology intersection matrix is constant if the period matrix and its dual have a C-form differential equation. For this it will be important to know when iterated integrals are linearly independent over the function field $\cF$. In the case of a single variable, $r=1$, the answer to this question was given in ref.~\cite{deneufchatel:hal-00558773}, where it was shown that the iterated integrals in eq.~\eqref{eq:iterated_int_def} are linearly independent if and only if $\mathbb{A}\cap\rd\cF=\{0\}$. Since this last condition is the defining property of a system in C-form (cf.~Definition~\ref{def:canon}), all iterated integrals that arise in the $\eps$-expansion of the path-ordered exponential in eq.~\eqref{eq:Pexp_def} are linearly independent over $\cF$, at least for $r=1$. 

In the case of several variables, $r>1$, the situation is less straightforward. The proof of ref.~\cite{deneufchatel:hal-00558773} relies on the existence of a basis for the specific
iterated integrals that appear in the $\eps$-expansion of the path-ordered exponential.
For $r>1$, the individual iterated integrals that arise in the $\eps$-expansion of the path-ordered exponential are in general not homotopy-invariant, but they depend on the details of the path $\gamma$. The flatness condition~\eqref{eq:flatness} ensures that only homotopy-invariant combinations of iterated integrals show up in the expansion.  
In the following, we want to argue that we can identify a basis for the homotopy-invariant combinations. Hence similarly to the single variable case we can conclude that all appearing combinations are linearly independent.

Let us start by discussing the flatness condition~\eqref{eq:flatness_canonical}. Expanding $\bA(\bx)$ into the basis of $\mathbb{A}$, we find
\beq\label{eq:Lie_1}
\sum_{i,j=1}^p[\bA_i,\bA_j]\,\omega_i\wedge\omega_j = 0\,.
\eeq
Let $\eta_k$ be a basis for the vector space $Z^2(\cA)$ of closed two-forms. Then we can expand $\omega_i\wedge\omega_j$ into that basis
\beq
\omega_i\wedge\omega_j = \sum_{k=1}^qa_{ijk}\eta_k\,,\quad a_{ijk}\in\KK\,.
\eeq
We note that, for $1\leq i,j \leq p$, we can always pick the basis $\eta_k$ such that $q$ is finite (because the $\omega_i\wedge\omega_j$ generate a finite-dimensional vector space).
Equation~\eqref{eq:Lie_1} then becomes
\beq
\sum_{k=1}^q\eta_k\sum_{i,j=1}^pa_{ijk}\,[\bA_i,\bA_j] = 0\,,
\eeq
and since the $\eta_k$ form a basis, we must have
\beq\label{eq:A_commutators}
\sum_{i,j=1}^pa_{ijk}\,[\bA_i,\bA_j] = 0\,,\qquad 1\le k\le q\,.
\eeq

In order to understand which combinations of iterated integrals are homotopy-invariant, it is useful to package the iterated integrals into a generating functional (cf.~eq.~\eqref{eq:Pexp_def_exp}):
\beq\label{eq:genfunc_original}
\mathbb{G} ={1} + \sum_{k=1}^\infty \sum_{1\le i_1,\ldots,i_k\le p}t_{i_1}\cdots t_{i_k}  I_{\gamma}(\omega_{i_1},\ldots,\omega_{i_k}) \,,
\eeq
where the $t_i$, $1\le i\le p$, are non-commuting variables that satisfy the same commutation relations as those satisfied by the matrices $\bA_i$ in eq.~\eqref{eq:A_commutators},\footnote{We will see in section~\ref{sec:Lie} that the matrices $\bA_i$ can be seen as a representation of some Lie algebra generated by the $t_i$.}
\beq\label{eq:t_commutators}
\sum_{i,j=1}^pa_{ijk}\,[t_i,t_j] = 0\,,\qquad 1\le k\le q\,.
\eeq
The generating functional $\mathbb{G}$ satisfies the differential equation
\beq\label{eq:gen_func_DEQ}
\rd\mathbb{G} = \left(\sum_{i=1}^pt_i\omega_i \right)\mathbb{G}\,,
\eeq
and the commutation relations in eq.~\eqref{eq:t_commutators} guarantee that the connection is flat. The variables $t_i$ generate a $\KK$-algebra ${A}_{\mathbb{A}}$, and identifying the combinations of iterated integrals that are homotopy-invariant is equivalent to identifying a basis for ${A}_{\mathbb{A}}$. It turns out that ${A}_{\mathbb{A}}$ belongs to the class of so-called \emph{quadratic algebras}, and for quadratic algebras a basis can be found in an algorithmic 
fashion. In the following, we only need the existence of a basis $\cB_{\bbA}$ for $A_{\bbA}$.
In appendix~\ref{app:DDMS_1} we review how to find such a basis. Having fixed a basis for $A_{\bbA}$, we can  express all the words $t_{i_1}\cdots t_{i_k}$ in terms of this basis, and   the generating functional can be written as
\beq \label{expgenfunc}
\mathbb{G} = \sum_{w\in \mathcal{B}_{\mathbb{A}}}w\,J(w)\,,
\eeq
where $J(w)$ denotes a combination of iterated integrals that is homotopy-invariant. The $\KK$-vector space of all homotopy-invariant combinations of iterated integrals one can form out of the letters from $\mathbb{A}$ is then
\beq\label{eq:VA_def}
V_{\mathbb{A}} := \left\langle J(w): w\in \mathcal{B}_{\mathbb{A}}\right\rangle_{\KK}\,.
\eeq

Having identified a basis for $A_{\bbA}$, in appendix~\ref{app:DDMS_2} we prove the following result:
\begin{theorem}\label{thm:lin_dep}
We assume the notations and conventions from above hold. Then, the following statements are equivalent:
\begin{enumerate}
\item $\mathbb{A}\cap \rd\cF = \{0\}$,
\item the functions $J(w)$, $w\in \mathcal{B}_{\mathbb{A}}$, are linearly independent over $\cF$.
\end{enumerate}
\end{theorem}
For the case $r=1$, this theorem reduces to the theorem from ref.~\cite{deneufchatel:hal-00558773}. To our knowledge, the extension of the theorem from ref.~\cite{deneufchatel:hal-00558773} to $r>1$ is new and has not appeared before. The proof of Theorem~\ref{thm:lin_dep} is essentially identical to the proof in ref.~\cite{deneufchatel:hal-00558773} for the case $r=1$, because the proof is purely algebraic. For completeness, we collect the prove in appendix~\ref{app:DDMS_2}.

Theorem~\ref{thm:lin_dep} has the following corollary:
\begin{corollary}\label{thm:main_1}
Consider two systems in C-form for the same algebra $\cA$ and of the same dimension $N$, 
\beq
\rd \bG_i(\bx,\eps) = \eps\bA^{(i)}(\bx)\bG_i(\bx,\eps)\,,\qquad i=1,2\,.
\eeq
Assume that there are matrices of full rank,
\beq
\bDelta(\bx,\eps) \in \GL(N,\cA_\eps) \textrm{~~~and~~~} \bH(\eps) \in \GL(N,\mathbb{C}(\eps))\,,
\eeq
such that
\beq\label{eq:trbr_proof}
\bG_1(\bx,\eps)\big( \bH(\eps)^{-1}\big)^T\bG_2(\bx,\eps)^T =\bDelta(\bx,\eps)\,.
\eeq
Then $\rd\bDelta(\bx,\eps)=0$.
\end{corollary}

\begin{proof}
We know that there are constant matrices $\bG_{0,i}(\eps)\in \KK^{N\times N}$ such that
\beq
\bG_i(\bx,\eps) = \bU_\gamma^{(i)}(\bx,\eps)\bG_{0,i}(\eps)\,,
\eeq
where the $\bU_\gamma^{(i)}(\bx,\eps)$ are path-ordered exponentials as in eq.~\eqref{eq:Pexp_def}. Then eq.~\eqref{eq:trbr_proof} takes the form
\beq\label{eq:riemann_proof_2}
\bDelta(\bx,\eps)  = \bU^{(1)}_{\gamma}(\bx,\eps)\big(\bH_{12}(\eps)^{-1}\big)^T\bU^{(2)}_{\gamma}(\bx,\eps)^T\,,
\eeq
with
\beq
 \bH_{12}(\eps)=\big(\bG_{0,1}(\eps)^{-1}\big)^T\bH(\eps)\bG_{0,2}(\eps)^{-1}\,.
\eeq

We may apply Theorem~\ref{thm:lin_dep}, which has the following immediate corollary. Consider the vector space $V_{\mathbb{C}} := \mathbb{C}\otimes_{\KK}V_{\mathbb{A}}$ (with $V_{\mathbb{A}}$ given in eq.~\eqref{eq:VA_def}). Then
any element of $V_{\mathbb{C}}$ is of the form $v=\sum_{w\in\mathcal{B}_{\mathbb{A}}} c_w J(w)$ with $c_w\in\mathbb{C}$ and $J(w) \in V_{\mathbb{A}}$. If $v$ is also in $\mathcal{F}_{\mathbb{C}}$, then from the linear independence of the $J(w)$ over $\mathcal{F}_\mathbb{C}$ it follows that $c_w = 0 $ for $w\neq \emptyset$ and consequently $v=c_\emptyset\in\mathbb{C}$. Thus:  
\beq
V_{\mathbb{C}} \cap \cF=\mathbb{C}\,.
\eeq

We expand $\bDelta(x,\eps)$ in $\eps$,
\beq
\Delta_{ij}(x,\eps) = \sum_{k=k_0}^\infty \eps^k\,\Delta_{ij}^{(k)}(x)\,,
\eeq
and we know that $\Delta_{ij}^{(k)}(x)\in \cA \subseteq \cF$. From looking at the right-hand side of eq.~\eqref{eq:riemann_proof_2}, we also know that $\Delta_{ij}^{(k)}(x)\in V_{\mathbb{C}}$.  Hence, $\Delta_{ij}^{(k)}(x)\in V_{\mathbb{C}}\cap\cF = \mathbb{C}$, and so $\Delta_{ij}^{(k)}(x)$ is constant. Because the only constants in $\cA$ lie in $\KK$, we have $\Delta_{ij}^{(k)}(x)\in \KK$, as claimed.
\end{proof}


\subsection{Interpretation for Feynman integrals}

Let us now discuss how the results from the previous section can be applied to Feynman integrals.
We assume that we have a vector of master integrals $\bI(\bx,\eps)$ obtained from IBP identities that satisfies the differential equation
\beq\label{eq:DEQ_sec4}
\rd\bI(\bx,\eps) = \bOmega(\bx,\eps)\bI(\bx,\eps) \,.
\eeq
We know that Feynman integrals are periods of a (relative) twisted cohomology theory, and a fundamental solution matrix $\bP(\bx,\eps)$ of eq.~\eqref{eq:DEQ_sec4} is the associated period matrix. The connection to twisted cohomology also allows one to define \emph{dual} master integrals $\bs{\check{I}}(\bx,\eps)$, which satisfy the differential equation (cf.~eq.~\eqref{eq:Gauss-Manin})
\beq\label{eq:dDEQ_sec4}
\rd\bIc(\bx,\eps) = \bOmegac(\bx,\eps)\bIc(\bx,\eps) \,.
\eeq
We can compute the intersection matrix $\bC(\bx,\eps)$ for this choice of master integrals and their duals. It will be a rational function of $\bx$, and it satisfies the differential equation~\eqref{eq:Gauss-Manin}. Note that the differential equations~\eqref{eq:DEQ_sec4} and~\eqref{eq:dDEQ_sec4} have essentially the same singularities (except for possibly additional ones stemming from the zeros of $\det(\bC)$), because the matrices defining the differential equations are related by eq.~\eqref{eq:Gauss-Manin}, which can be cast in the form,
\beq\label{eq:Omega_del_relation}
\bOmegac(\bx,\eps)^T = \bC(\bx,\eps)^{-1}\Big(\bOmega(\bx,\eps)\bC(\bx,\eps) - \rd\bC(\bx,\eps)\Big)\,,
\eeq
and the entries of $\bC(\bx,\eps)$ are rational functions. 

We assume that we can find canonical bases for both the differential equation~\eqref{eq:DEQ_sec4} and its dual~\eqref{eq:dDEQ_sec4}, i.e., we assume that we can find changes of basis,
\beq\bsp
\bI(\bx,\eps) = \bS(\bx,\eps)\bJ(\bx,\eps)\,,\\
\bIc(\bx,\eps) = \bSc(\bx,\eps)\bJc(\bx,\eps)\,,
\esp\eeq
such that the differential equations are in $\eps$-factorised form,
\beq\bsp\label{eq:canon_eqs}
\rd\bJ(\bx,\eps) &\,= \eps\bA(\bx)\bJ(\bx,\eps) \,,\\
\rd\bJc(\bx,\eps) &\,= \eps\bAc(\bx)\bJc(\bx,\eps) \,.
\esp\eeq
Differential equations in canonical form for dual integrals were studied in refs.~\cite{Caron-Huot:2021xqj,Giroux:2022wav,De:2023xue}. We assume that the differential equations~\eqref{eq:canon_eqs} are in C-form (we recall that this is the case for all known examples). 
The entries of $\bA(\bx)$ and $\bAc(\bx)$ allow us to define the algebra $\cA$. 

The intersection matrix computed with the canonical bases is related to the one for the original bases via
\beq
\bDelta(\bx,\eps) = \bS(\bx,\eps)^{-1}\bC(\bx,\eps)\left(\bSc(\bx,\eps)^{-1}\right)^T\,.
\eeq
Note that all entries of $\bDelta(\bx,\eps)$ lie in $\cA_{\eps}$, and that $\bDelta(\bx,\eps)$ satisfies the differential equation
\beq\label{eq:dDelta_dual}
\rd\bDelta(\bx,\eps) = \eps\bA(\bx)\bDelta(\bx,\eps) + \eps\bDelta(\bx,\eps)\bAc(\bx)^T\,.
\eeq
Finally the TRBRs from eq.~\eqref{generalriemann} take the form
\beq
(2\pi i)^{-n}\,\bG(\bx,\eps)\left(\bH(\eps)^{-1}\right)^T\bs{\check{G}}(\bx,\eps)^T = \bDelta(\bx,\eps)\,,
\eeq
where we defined
\beq\bsp
\bG(\bx,\eps) &\,= \bS(\bx,\eps)^{-1}\bP(\bx,\eps)\,,\\
\bs{\check{G}}(\bx,\eps) &\,= \bSc(\bx,\eps)^{-1}\bPc(\bx,\eps)\,.
\esp\eeq
We then see that Corollary~\ref{thm:main_1} applies with $\bA_1(\bx) = \bA(\bx)$ and $\bA_2(\bx) = \bAc(\bx)$, and we conclude that $\rd\bDelta(\bx,\eps)=0$. In other words, we see that after passing to a canonical basis for both the Feynman integrals and their duals, the intersection matrix becomes constant in $\bx$. We can in fact say more about the structure of the intersection matrix:
\begin{theorem}\label{thm:duality}
The matrix $\bDelta(\eps)$ can be written as a product
\beq
\bDelta(\eps) = \bDelta_0\widetilde{\bDelta}(\eps)\,,
\eeq
where $\bDelta_0\in\GL(N,\KK)$ is a constant matrix of full rank and $\widetilde{\bDelta}(\eps)\in\GL(N,\KK(\eps))$ is a matrix of full rank that commutes with $\bAc(\bx)^T$.
\end{theorem}
\begin{proof}
Since $\rd\bDelta=0$, eq.~\eqref{eq:dDelta_dual} implies:
\beq\label{eq:eqeqmatrix}
\bA(\bx)\bDelta(\eps) = -\bDelta(\eps)\bAc(\bx)^T\,.
\eeq
We write $\bA(\bx)$ and $\bAc(\bx)$ in the basis of $\mathbb{A}$ and we expand $\bDelta(\eps)$ in $\eps$,
\beq
\bDelta(\eps) = \eps^m\,\sum_{k=0}^{\infty}\bDelta_k\eps^k\,,
\eeq
for some integer $m$.
Then eq.~\eqref{eq:eqeqmatrix} is equivalent to the set of equations
\beq\label{eq:eqeq2}
\bA_i\bDelta_k =-\bDelta_k\bAc_i^T\,,\qquad 1\le i\le p\,, k\ge 0\,.
\eeq
The Laurent coefficients $\bDelta_k$ do not need to have full rank. Since $\bDelta(\eps)$ has full rank for generic $\eps$, there is some $\eps_0\in\mathbb{C}$ such that $\bDelta_0:=\bDelta(\eps_0)$ has full rank. We may choose $\eps_0\in\mathbb{Q}$. Indeed, the values of $\eps$ for which $\bDelta(\eps)$ does not have full rank are those for which $\det\bDelta(\eps)$ is either zero or infinite. $\det\bDelta(\eps)$ is a rational function of $\eps$, so all zeroes and poles are isolated. This implies that there is a neighborhood $U$ of $\eps=0$ such that $\bDelta(\eps)$ has full rank for all $\eps\in U\setminus\{0\}$. The set $U\setminus\{0\}$ contains an infinite number of rational numbers, and so we can pick $\eps_0$ to be one of those. Since $\bDelta(\eps) \in \KK(\eps)^{N\times N}$, this implies that all entries of $\bDelta_0$ lie in $\KK$. 

Letting $\eps=\eps_0$, we see that eq.~\eqref{eq:eqeqmatrix} takes the form
\beq\label{eq:dual_intertwine}
\bA(\bx) =-\bDelta_0\bAc(\bx)^T\bDelta_0^{-1}\,.
\eeq
Inserting this relation back into eq.~\eqref{eq:eqeqmatrix} we find, with $\widetilde{\bDelta}(\eps):= \bDelta_0^{-1}\bDelta(\eps)$,
\beq
\widetilde{\bDelta}(\eps)\bAc(\bx)^T=\bAc(\bx)^T\widetilde{\bDelta}(\eps)\,.
\eeq
Hence, we see that $\widetilde{\bDelta}(\eps)$ and $\bAc(\bx)^T$ commute.
\end{proof}

\subsection{Interpretation in terms of Lie algebras}
\label{sec:Lie}
In this section we discuss how we can associate a Lie algebra to every system in canonical form or in C-form, and we give an interpretation of the matrix $\bDelta_0$ in that context.

Consider  the Lie algebra $\mathfrak{g}_{\bbA}$ obtained as the quotient of the free Lie algebra over $\KK$ with generators $t_i$, $1\le i\le p$ and the Lie ideal generated by (cf. eq.~\eqref{eq:t_commutators})
\beq
\sum_{i,j=1}^pa_{ijk}\,[t_i,t_j]\,,\qquad 1\le k\le q\,.
\eeq
It is easy to see that the quadratic algebra $A_{\mathbb{A}}$ defined in section~\ref{sec:line_dep} is the universal enveloping algebra of $\mathfrak{g}_{\bbA}$. Note that if there is a single variable ($r=1$), the flatness condition is trivially satisfied, and the Lie algebra $\mathfrak{g}_{\bbA}$ is a free Lie algebra.

The matrices $\bA_i$ provide an $N$-dimensional representation of $\mathfrak{g}_{\bbA}$,
\beq
\rho: \mathfrak{g}_{\bbA}\to \End(\mathbb{C}^N)\,, \quad t_i \mapsto \bA_i\,.
\eeq
We see that we can associate to $\bbA$ a Lie algebra $\mathfrak{g}_\bbA$, and every system in C-form provides a (in general reducible) representation. Said differently, all possible systems in C-form can be characterised by the pair $(\mathfrak{g}_{\bbA}, \rho)$. In particular, this implies that we can associate a Lie algebra to every family of Feynman integrals for which we can find a set of differential equations in C-form. Note that many Feynman integrals share the same Lie algebra, but only differ by the choice of the representation $\rho$. For example, many Feynman integrals can be expressed in terms of harmonic polylogarithms~\cite{Remiddi:1999ew}. In that case $\cA$  is the ring of rational functions with poles at most at $x\in\{0,\pm1,\infty\}$,
\beq
\cA = \mathbb{Q}\left[x,\tfrac{1}{x},\tfrac{1}{1-x},\tfrac{1}{1+x}\right]\,,
\eeq
and $\bbA$ is generated by logarithmic forms,
\beq
\bbA = \left\langle\dlog x,\dlog(1-x),\dlog(1+x)\right\rangle_{\mathbb{Q}}\,.
\eeq
In this case the Lie algebra $\mathfrak{g}_{\bbA}$ is the free Lie algebra in 3 generators. All Feynman integrals that evaluate to harmonic polylogarithms share this same Lie algebra, but they differ by the representation $\rho$ which defines the matrices that appear in the differential equation.


Since a family of Feynman integrals and their duals share the same algebra $\cA$ (cf. eq.~\eqref{eq:Omega_del_relation}, which relates the differential equation matrices $\bOmega(\bx,\eps)$ and $\bOmegac(\bx,\eps)$), they correspond to different representations of the same Lie algebra $\mathfrak{g}_{\bbA}$,
\beq
\rho(t_i) = \bA_i \textrm{~~~and~~~} \check{\rho}(t_i) = \bs{\check{A}}_i\,.
\eeq
We now discuss the connection between these two representations. First, we recall that for every representation of a Lie algebra, there is a dual representation $\rho^* = -\rho^T$. Equation~\eqref{eq:dual_intertwine} implies that the representation $\check{\rho}$ is equivalent to the dual representation $\rho^*$:
\beq\label{eq:dual_matrix}
\check{\rho}(t_i) = \bAc_i = -\bDelta_0^T\bA_i^T\big(\bDelta_0^{-1}\big)^T = \bDelta_0^T\rho^*(t_i)\big(\bDelta_0^T\big)^{-1}\,,
\eeq
where we recall that two representations $\rho_1$ and $\rho_2$ are equivalent if there is an invertible matrix $\bM$ such that $\rho_2=\bM\rho_1\bM^{-1}$.

Having established how we can associate a Lie algebra representation to a system in C-form, one can study what one may learn from representation theory about such systems. An important class of Lie algebra representations are \emph{irreducible} representations. We recall that a Lie algebra representation $\rho$ is called \emph{reducible} is there is a matrix $\bM\in\GL(N,\mathbb{C})$ such that $\bM\rho\bM^{-1}$ is block upper-triangular, and an irreducible representation is one that is not reducible. We call a system irreducible if the associated representation of $\mathfrak{g}_{\bbA}$ is irreducible. An important tool to characterise irreducible representations is Schur's lemma, and we can indeed formulate a variant of Schur's lemma for irreducible systems in C-form.
\begin{lemma}\label{lemma:Schur}
Let $\bA(\bx)$ be the matrix describing an irreducible system in C-form, and $\bM(\eps) \in \KK({\eps})^{N\times N}$. Then $\bA(\bx)$ and $\bM(\eps)$ commute if and only if $\bM(\eps)$ is a multiple of the identity, i.e., there is a rational function $f\in \KK({\eps})$ such that $\bM(\eps) = f(\eps)\mathds{1}$.
\end{lemma}
\begin{proof}
We decompose $\bA(\bx)$ into the basis $\omega_i$, and we expand $\bM(\eps)$ in $\eps$,
\beq
\bM(\eps) = \sum_{k=k_0}^{\infty}\bM_k\eps^k\,.
\eeq
Then $\bA(\bx)$ and $\bM(\eps)$ commute if and only if
\beq
[\bA_i,\bM_k] = 0\,,\qquad 1\le i\le p\,, \quad k\ge k_0\,.
\eeq
%
%
Since the system is assumed irreducible, the representation $\rho(t_i)=\bA_i$ from section~\ref{sec:Lie} is irreducible over $\mathbb{C}$. By Schur's lemma, every operator that commutes with a representation of a Lie algebra that is irreducible over an algebraically closed field must be proportional to the identity, and so $\bM_k = \lambda_k\mathds{1}$ for some $\lambda_k\in\KK$. Then the claim follows.
\end{proof}
Let us discuss an important implication of Schur's lemma for irreducible systems. 
It is easy to see that, if $\rho$ is irreducible, then so is its dual $\rho^*$. Theorem~\ref{thm:duality} combined with Schur's lemma then implies that $\bDelta(\eps) = f(\eps)\bDelta_0$, for some rational function $f$. Note that the representations obtained from Feynman integrals are typically not irreducible, because $\bA(\bx)$ has a block-diagonal structure. The blocks on the diagonal, which describe maximal cuts, however, may be irreducible, and then Schur's lemma predicts that the intersection matrix for maximal cuts in a canonical basis takes a particularly simple form.

Another interesting class of representations are self-dual representations, i.e., representations for which $\rho$ and $\rho^*$ are equivalent. In the next section we will argue that self-dual representations naturally arise from maximal cuts in a canonical basis, and we study the consequences of this.


\section{Maximal cuts and self-dual systems}
\label{sec:main:thm}

\subsection{Self-dual systems}

We now focus on maximal cuts, we assume that we work in a canonical basis for both the Feynman integrals and their duals, and that the system is at the same time in C-form. As noted before, all currently known examples in canonical form are also in C-form. At this point we cannot exclude that there might be examples of canonical forms that are not C-forms or vice versa, and we therefore  specify both conditions separately.
We use the notation from section~\ref{sec:maximal-cuts}. Equation~\eqref{dualomegaepsilon} implies the relation $\eps\bAc^{\textrm{m.c.}}(\bx) = -\eps\bA^{\textrm{m.c.}}(\bx)$, or equivalently,
\beq
\bAc^{\textrm{m.c.}}(\bx) = -\bA^{\textrm{m.c.}}(\bx)\,.
\eeq
Inserting this relation into eq.~\eqref{eq:dual_matrix}, we see that the representation of $\mathfrak{g}_{\bbA}$ is self-dual, and we have
\beq\label{eq:sd_condition_1}
\bA^{\textrm{m.c.}}(\bx) = \bDelta_0^T\bA^{\textrm{m.c.}}(\bx)^T\big(\bDelta_0^{-1}\big)^T\,.
\eeq
If we assume in addition that the representation associated to $\bA^{\textrm{m.c.}}(\bx) $ is irreducible,\footnote{If this condition fails, then the statement holds for all irreducible blocks individually.} then we have the following result:
\begin{theorem}\label{thm:main_sd}
If the representation of $\mathfrak{g}_{\bbA}$ is both self-dual and irreducible, then we have $\bDelta(\eps) = f(\eps)\,\bDelta_0$, where $f\in\KK(\eps)$ is a rational function of $\eps$ and $\bDelta_0\in \GL(N,\KK)$ is either symmetric or antisymmetric. 
\end{theorem}
\begin{proof}
From Schur's lemma, we know already that $\bDelta(\eps) = f(\eps)\,\bDelta_0$. It remains to show that $\bDelta_0$ has the required symmetry. This can be seen as follows: For a self-dual system, eq.~\eqref{eq:sd_condition_1} must hold. We see eq.~\eqref{eq:sd_condition_1} as an equation for $\bDelta_0$. This equation is clearly linear in $\bDelta_0$, and if $\bDelta_0$ is a solution, so is $\bDelta_0^T$. Since it follows from Schur's lemma that the solution is unique up to multiplication by some constant, we must have
\beq
\bDelta_0^T = k\, \bDelta_0\,,
\eeq
for some $k\in\mathbb{K}$. By taking the transpose, we can determine the value of $k$ to be $\pm1$:
\beq
\bDelta_0 = k\, \bDelta_0^T =k^2\,\bDelta_0 \,.
\eeq
\end{proof}

\begin{examplenew}{1} \label{ex2f12}
Let us illustrate Theorem~\ref{thm:main_sd} for the hypergeometric function introduced in example \ref{ex2f1}.
The relevant twist is given by
\begin{equation}
    \Phi_T=z^{-\tfrac{1}{2}+a \varepsilon}(1-z)^{-\tfrac{1}{2}+b \varepsilon}(1-x z)^{-\tfrac{1}{2}+c \varepsilon} \,.
\end{equation}
A basis for the cohomology group $H^1(X_T,\nabla)$, with $X_T=\mathbb{C}-\{0,1,\tfrac{1}{x},\infty\}$,
is 
\beq
\varphi_{T,1}=\rd z \textrm{~~~and~~~} \varphi_{T,2}=\rd z\, z\,.
\eeq
Choosing the dual basis such that $\check{\bP}_{T}(x,\eps)=\bP_{T}(x,-\varepsilon)$ leads to the cohomology intersection matrix
\begin{align}
\bC_{T}(x,\eps)=\left(
\begin{array}{cc}
 0 & \frac{2}{x (2 \varepsilon  (a+b+c)-1)} \vspace{0.2cm}\\
 \frac{2}{2 x \varepsilon  (a+b+c)+x} & \frac{4 \varepsilon  (x (a+c)+a+b)}{x^2 \left(4 \varepsilon ^2 (a+b+c)^2-1\right)} \\
\end{array}
\right)\,.
\end{align}
If we change basis to the canonical basis using the matrix $\bS_T(x,\varepsilon)$ in eq.~\eqref{eq:ST_def}, we obtain 
\beq\label{delta2F1}
\bDelta_T(x,\eps) = \bS_T(x,\eps)^{-1}\bC_T(x,\eps)\left(\bS_T(x,-\eps)^{-1}\right)^T = f_T(\eps)\bDelta_{0,T}\,,
\eeq
with
\beq
f_T(\eps) = \frac{2}{\varepsilon} \textrm{~~~and~~~} \bDelta_{0,T}=\left(
\begin{smallmatrix}
 0 & \phantom{-}1 \\
 1 & -2 \\
\end{smallmatrix}
\right)\,.
\eeq
\end{examplenew}
\begin{examplenew}{2}
\label{exSunrise2}
    Let us also discuss the example of the unequal-mass sunrise integral from example \ref{exSunrise}. The cohomology intersection matrix $\bC_\mathrm{S}(\bx,\eps)$ for the maximal cuts of the unequal-mass sunrise integral was computed in ref.~\cite{Duhr:2024rxe}. If we perform the transformation to the canonical basis  described in ref.~\cite{Bogner:2019lfa}, we find    
    \beq\label{deltaSunrise}
\bDelta_\mathrm{S}(\bx,\eps) = \bS_\mathrm{S}(\bx,\eps)^{-1}\bC_\mathrm{S}(\bx,\eps)\left(\bS_\mathrm{S}(\bx,-\eps)^{-1}\right)^T = f_\mathrm{S}(\eps)\bDelta_{0,\mathrm{S}}\,,
\eeq
with
\begin{equation}
    \label{eq:intMatrixSunrise}
       f_\mathrm{S}(\eps)=\frac{\eps^7}{1-4\eps^2} \textrm{~~and~~} \bDelta_{0,\mathrm{S}}=\begin{pmatrix}
            0 & 0 & 0 & -\frac{1}{8} \\
            0 & -3 & 0 & 0 \\
            0 & 0 & -1 & 0 \\
            -\frac{1}{8} & 0 & 0 & 0
        \end{pmatrix} \,.
    \end{equation}
    As expected, this matrix does not depend on the kinematic parameters $\bx$ anymore and can be written as a rational function $f_\mathrm{S}(\eps)$ times a constant matrix $\bDelta_{0,\mathrm{S}}$.
\end{examplenew}


\subsection{Symmetries of differential equations from self-duality}
\label{sec:symmetries}
Since $\bDelta_0$ must be either symmetric or antisymmetric, we can cast eq.~\eqref{eq:sd_condition_1} in the form
\beq\label{eq:sd_condition_2}
\bA^{\textrm{m.c.}}(\bx) = \bDelta_0\bA^{\textrm{m.c.}}(\bx)^T\bDelta_0^{-1}\,.
\eeq
Equation~\eqref{eq:sd_condition_2} requires that not all entries of $\bA^{\textrm{m.c.}}(\bx)$ are independent, and there must be relations among them.
We can thus interpret eq.~\eqref{eq:sd_condition_2} as self-duality and irreducibility imposing some symmetries on the matrix $\bA^{\textrm{m.c.}}(\bx)$. This implies precisely  the content of the observation of self-duality of ref.~\cite{Pogel:2024sdi}, though in that case only a very specific symmetry is allowed. 
Indeed, we see that eq.~\eqref{eq:sd_condition_2} agrees with eq.~\eqref{eq:weinzierl_sd} for $\bDelta_0 = \bK_N$. The matrix $\bDelta_0$ in eq.~\eqref{eq:sd_condition_2} can at least a priori be more general. In particular, it may be symmetric or antisymmetric, while $\bK_N$ is always symmetric. In the following we investigate the two cases in turn. We put $\KK=\mathbb{Q}$ from here on, in order to simplify the discussion, and because this is the typical case encountered for Feynman integrals. The extension of the discussion to algebraic number fields is not difficult.

\paragraph{The symmetric case.} We start by discussing the case $\bDelta_0^T=\bDelta_0$. $\bDelta_0$ is real and symmetric for $\KK=\mathbb{Q}$, and thus diagonalisable. In other words, there is a real orthogonal matrix $\bR_1$ and a diagonal matrix $\bD = \diag(\lambda_1,\ldots,\lambda_N)$ such that
\beq
\bDelta_0 = \bR_1\bD\bR_1^T = \widetilde{\bR}_1\widetilde{\bR}_1^T\,,
\eeq
with $\widetilde{\bR}_1 := \bR_1\sqrt{\bD}$ and $\sqrt{\bD}= \diag(\sqrt{\lambda_1},\ldots,\sqrt{\lambda_N})$\,. Note that, since $\bDelta_0$ has full rank, all its eigenvalues $\lambda_i$ are non-zero, and so $\widetilde{\bR}_1$ also has full rank.

If $\bG(\bx,\eps)$ denotes the period matrix in a canonical basis, we can change basis to
\beq
\bG'(\bx,\eps) = \widetilde{\bR}_1^{-1}\bG(\bx,\eps)\,.
\eeq
Since $\widetilde{\bR}_1$ is constant, we have changed basis to another canonical basis. In this new basis the intersection matrix takes the very simple form
\beq
\bDelta'(\eps) = \widetilde{\bR}_1^{-1}\bDelta(\eps)\left(\widetilde{\bR}_1^{-1}\right)^T = f(\eps)\mathds{1}\,,
\eeq
and the matrix of the differential equation is
\beq
\widetilde{\bA}^{\textrm{m.c.}}(\bx) =  \widetilde{\bR}_1^{-1}{\bA}^{\textrm{m.c.}}(\bx) \widetilde{\bR}_1\,,
\eeq
and eq.~\eqref{eq:sd_condition_2} implies that $\widetilde{\bA}^{\textrm{m.c.}}(\bx)$ is symmetric. This statement seems to be at odds with the observation of self-duality of ref.~\cite{Pogel:2024sdi}, which states that there is a basis in which the matrix of the differential equation becomes persymmetric, cf.~eq.~\eqref{eq:weinzierl_sd}. This apparent conundrum is solved as follows. Using exactly the same argument as before, we can show that for every real $\bM$ that is symmetric and has full rank, there is constant matrix $\widetilde{\bR}_M$ such that after rotation to this basis we have 
\beq
\bDelta'_0 = \bM \textrm{~~~and~~~} \widetilde{\bA}^{\textrm{m.c.}}(\bx) = \bM\widetilde{\bA}^{\textrm{m.c.}}(\bx)^T\bM^{-1}\,.
\eeq
This shows that, whenever $\bDelta_0$ is symmetric, there is a constant matrix $\widetilde{\bR}_K$ such that in this basis
\beq\label{eq:weinzierl_sd_2}
\bDelta'_0 = \bK_N \textrm{~~~and~~~} \widetilde{\bA}^{\textrm{m.c.}}(\bx) = \bK_N\widetilde{\bA}^{\textrm{m.c.}}(\bx)^T\bK_N^{-1}\,.
\eeq
This is precisely the statement of the observation of self-duality of ref.~\cite{Pogel:2024sdi}, cf.~eq.~\eqref{eq:weinzierl_sd}. In fact, it is easy to see that the observation of self-duality of ref.~\cite{Pogel:2024sdi} is entirely equivalent to the following statement:
\begin{quote}
\emph{For maximal cuts the matrix $\bDelta_0$ is always symmetric (and never antisymmetric).}
\end{quote}
This equivalent formulation of the observation of ref.~\cite{Pogel:2024sdi} has the advantage that it does not require a specific basis choice, but it holds regardless of the choice of canonical basis. We will come back to this point in section~\ref{sec:galois}. We know that the general formalism predicts that $\bDelta_0$ could be either symmetric or antisymmetric. To see if the observation of self-duality of ref.~\cite{Pogel:2024sdi} could be true generally, we need to investigate the antisymmetric case.

\paragraph{The antisymmetric case.} We now turn to the case $\bDelta_0^T=-\bDelta_0$. We start by observing that this case can only arise for an even number $N$ of master integrals in the top sector. Indeed, we must have
\beq
\det\bDelta_0 = \det\bDelta_0^T = (-1)^N\det\bDelta_0\,,
\eeq
and so $\det\bDelta_0=0$ for $N$ odd. Since we know that $\bDelta_0$ must have full rank, this limits this case to $N$ even. In particular, this implies that the observation of self-duality of ref.~\cite{Pogel:2024sdi} is correct generally for all cases with an odd number of master integrals in the top sector.

For $N$ even, there is always a real orthogonal matrix $\bs{Q}_1$ such that
\beq
\bDelta_{0} = \bs{Q}_1\bJ_N\bs{Q}_1^{-1}\,,
\eeq
where $\bJ_N$ is the standard symplectic matrix
\beq
\bJ_N = -\bJ_N^T = -\bJ_N^{-1} = \left(\begin{smallmatrix} \bs{0} & \mathds{1}\\ -\mathds{1} & \bs{0}\end{smallmatrix}\right)\,.
\eeq
It follows that there is a canonical basis $\bG'(\bx,\eps) = \widetilde{\bs{Q}}_1^{-1}\bG(\bx,\eps)$ with 
\beq\label{eq:symp_sym}
\bDelta'_0 = \bJ_N \textrm{~~~and~~~} \widetilde{\bA}^{\textrm{m.c.}}(\bx) = \bJ_N\widetilde{\bA}^{\textrm{m.c.}}(\bx)^T\bJ_N^{-1}\,.
\eeq
Equation~\eqref{eq:symp_sym} defines a new potential symmetry for the matrix of differential equations, which is not captured by the observation of self-duality of ref.~\cite{Pogel:2024sdi}. Note that for $N=2$, eq.~\eqref{eq:symp_sym} implies $\widetilde{\bA}^{\textrm{m.c.}}(\bx) = \diag(\omega,\omega)$, which is impossible if the system is irreducible.
\begin{examplenew}{1}
We can easily observe the symmetry in eq.~\eqref{eq:sd_condition_2} for the hypergeometric function from example~\ref{ex2f12}. For the matrices defined in eq.~\eqref{A2f1} and eq.~\eqref{delta2F1} we find
\begin{align}\label{eq:2F1_symmetry}
\bA_{T}(\tau)=\bDelta_{T,0}\bA^T_{T}(\tau)\bDelta_{T,0}^{-1}\,.
\end{align}
Equation~\eqref{eq:2F1_symmetry} can be interpreted as a constraint on the entries of $\bA_{T}(\tau)$. The most general matrix $\bA_{T}(\tau)$ that satisfies eq.~\eqref{eq:2F1_symmetry} has the form
\beq
\bA_{T}(\tau) = \begin{pmatrix} a(\tau) & b(\tau)\\c(\tau) & a(\tau)-2b(\tau)\end{pmatrix}\,.
\eeq
In particular, we see that the solution space is three-dimensional, and it is easy to check that this three-dimensional solution space is spanend by the three matrices $\bA_{T,i}$ in eq.~\eqref{singlea2f1}.
%
\end{examplenew}
\begin{examplenew}{2}
\label{exSunrise3}
%
Using eq.~\eqref{eq:intMatrixSunrise}, we see that the differential equation matrix for the maximal cuts satisfies the condition:
    \begin{equation}
    \label{eq:symmetryEqSunrise}
        \bA_\mathrm{S}^{\text{m.c.}}(\bx)=\bDelta_{0,\mathrm{S}}\bA_\mathrm{S}^{\text{m.c.}}(\bx)^T\bDelta_{0,\mathrm{S}}^{-1} \, .
    \end{equation}
We can again interpret eq.~\eqref{eq:symmetryEqSunrise} as imposing $\mathbb{Q}$-linear relations between the entries of the differential equation matrix $\bA_\mathrm{S}^{\text{m.c.}}(\bx)$. Explicitly we find the six relations
    \beq\bsp
        \left(\bA_\mathrm{S}^{\text{m.c.}}(\bx)\right)_{12}&=\frac{1}{24}\left(\bA_\mathrm{S}^{\text{m.c.}}(\bx)\right)_{24}\,, \\
        \left(\bA_\mathrm{S}^{\text{m.c.}}(\bx)\right)_{23}&=3\left(\bA_\mathrm{S}^{\text{m.c.}}(\bx)\right)_{32}\,, \\
        \left(\bA_\mathrm{S}^{\text{m.c.}}(\bx)\right)_{13}&=\frac{1}{8}\left(\bA_\mathrm{S}^{\text{m.c.}}(\bx)\right)_{34}\,, \\
        \left(\bA_\mathrm{S}^{\text{m.c.}}(\bx)\right)_{31}&=8\left(\bA_\mathrm{S}^{\text{m.c.}}(\bx)\right)_{43}\,, \\
        \left(\bA_\mathrm{S}^{\text{m.c.}}(\bx)\right)_{11}&=\left(\bA_\mathrm{S}^{\text{m.c.}}(\bx)\right)_{44}\,, \\
        \left(\bA_\mathrm{S}^{\text{m.c.}}(\bx)\right)_{21}&=24\left(\bA_\mathrm{S}^{\text{m.c.}}(\bx)\right)_{42} \,,
    \esp\eeq
    These can easily be checked to hold for the differential equation matrix given in ref.~\cite{Bogner:2019lfa}. These relations allow 10 linearly independent solutions, implying in particular that the dimension of $\mathbb{A}_\mathrm{S}$ can not be higher than that.
\end{examplenew}

\paragraph{Conclusion.} Our analysis shows that it is always possible to find a basis in which the matrix  of the differential equation satisfies the symmetry properties in eq.~\eqref{eq:weinzierl_sd_2} or eq.~\eqref{eq:symp_sym}. While the former corresponds to the statement of the observation of self-duality from ref.~\cite{Pogel:2024sdi}, the latter is not captured by that observation. Hence, we conclude that the observation holds for $N=2$ and all odd values of $N$, but there is the possibility that the self-duality fails starting for $N=2m\ge 4$. We note here that the observations of ref.~\cite{Pogel:2024sdi} rely on the explicit analysis of maximal cuts with $N\le 4$ and the equal-mass banana integrals. 
We thus see two options: either the observation of ref.~\cite{Pogel:2024sdi} holds, and $\bDelta_0$ is always symmetric, calling for a deeper explanation, or the examples considered in ref.~\cite{Pogel:2024sdi} were insufficient to capture the full breath of the complexity that may arise from Feynman integrals.


\section{Some comments on Galois symmetries}
\label{sec:galois}

In ref.~\cite{Pogel:2024sdi} it was pointed out that, when the basis of master integrals is rotated such that eq.~\eqref{eq:weinzierl_sd} is satisfied, square roots are introduced and there may be Galois symmetries relating the integrals. In the previous section, we have shown that the notion of self-duality of ref.~\cite{Pogel:2024sdi} can be recast in a basis-invariant manner, and that there is nothing special about the basis which manifests the symmetry in eq.~\eqref{eq:weinzierl_sd}. 
In this section, we review how these square roots arise.

\subsection{A very short review of Galois symmetry}
Galois symmetry is, very loosely speaking, the symmetry relating different solutions of a polynomial equation. More precisely, consider for example a polynomial $P(x)$ with rational coefficients,
\beq
P(x) = \sum_{i=1}^nc_n\,x^n\,,\qquad c_i\in \mathbb{Q}\,.
\eeq
While the coefficients of $P$ live in $\mathbb{Q}$, the roots of $P$ may live in a field extension, i.e., in a field $\KK$ containing $\mathbb{Q}$. Then, if $\sigma$ is a field automorphism of $\KK$ that fixes the elements of $\mathbb{Q}$ (meaning that $\sigma$ is a map from $\KK$ to $\KK$ that preserves the addition and the multiplication, and $\sigma(x)=x$ for all $x\in\mathbb{Q}$), then $\sigma$ transforms one root into another. Indeed, we have for $x$ a root:
\beq
P(\sigma(x)) =\sum_{i=1}^nc_n\,\sigma(x)^n = \sigma(P(x))  = \sigma(0) = 0\,.
\eeq

Galois symmetry is most easily illustrated for quadratic polynomials. Consider the polynomial
\beq
P(x) = a\,x^2+b\,x+c\,,\qquad a,b,c\in\mathbb{Q}\,,\quad a\neq 0\,.
\eeq
Its roots are
\beq
x_{\pm} = \frac{-b\pm \sqrt{\Delta}}{2a}\,,\qquad \Delta=b^2-4ac\,.
\eeq
If the discriminant $\Delta$ is the square of a rational number, then both roots of $P$ are rational. However, if $\Delta$ is not the square of any rational number, then the roots do not lie in $\mathbb{Q}$. Instead, they lie in the larger field
\beq
\mathbb{Q}\big(\sqrt{\Delta}\big) = \left\{\alpha+\beta\sqrt{\Delta}:\alpha,\beta\in\mathbb{Q}\right\}\,.
\eeq
The Galois group is in this case the group $\mathbb{Z}_2 = \{\textrm{id},\sigma\}$, which acts on $\mathbb{Q}(\sqrt{\Delta})$ by changing the sign of the root,
\beq
\sigma(\alpha+\beta\sqrt{\Delta}) = \alpha-\beta\sqrt{\Delta}\,,
\eeq
and so in particular it exchanges the two roots, 
\beq
\sigma(x_{\pm}) = x_{\mp}\,.
\eeq

\subsection{Galois symmetries from self-duality}
Reference~\cite{Pogel:2024sdi} observed that the rotation $\bM$ that relates ${\bA}^{\textrm{m.c}}(\bx)$ to $\widetilde{\bA}^{\textrm{m.c.}}(\bx)$ so that $\widetilde{\bA}^{\textrm{m.c.}}(\bx)$ is persymmetric (cf.~eq.~\eqref{eq:weinzierl_sd}) involves additional square roots (in particular the imaginary unit $i$ and $\sqrt{3}$ for the examples considered in ref.~\cite{Pogel:2024sdi}). As a consequence, Galois symmetries appear which relate different master integrals and different entries in the matrix $\widetilde{\bA}^{\textrm{m.c.}}(\bx)$. 

We again fix $\KK=\mathbb{Q}$ in the following for simplicity, and we assume that we are in the case where $\bDelta_0$ is symmetric (because otherwise the observation of self-duality of ref.~\cite{Pogel:2024sdi} would not hold in the first place). Then we know from Theorem~\ref{thm:main_sd} that all entries of $\bDelta_0$ are rational, and the matrix of the differential equation satisfies the symmetry property in eq.~\eqref{eq:sd_condition_2}.
The important point to note is that the most natural symmetry property is defined in terms of the intersection matrix, which is equal to $\bDelta_0$ up to an overall factor, and that the entries of $\bDelta_0$ are all rational. In other words, the symmetry properties of the differential equation matrix are naturally defined over $\mathbb{Q}$, and no additional square roots are needed.

Of course, one may change basis so that $\bDelta_0$ takes a different form, and we have seen that we can find a basis where we can fix $\bDelta_0$ to be any given real and symmetric matrix. In particular, there is a basis where $\bDelta_0$ becomes equal to $\bK_N$, cf.~eq.~\eqref{eq:weinzierl_sd_2}. As described in section~\ref{sec:symmetries}, the matrix $\widetilde{\bR}_K$ that describes the change of basis is obtained by diagonalising a matrix. It is well known that, if a matrix has rational entries, the eigenvalues and eigenvectors may not be rational, but they may be algebraic. Hence, by insisting to express the self-duality in a basis where $\bDelta'_0 = \bK_N$ as in eq.~\eqref{eq:weinzierl_sd_2}, we may end up introducing spurious square roots. Since the eigenvalues and eigenvectors are obtained by solving polynomial equations, they will give rise to the Galois symmetries related to the relevant field extension. While the discussion here is limited to the maximal cuts, the same square roots, and relations induced by the Galois symmetry that change the sign of the roots, also appear in the part of the differential equation matrix describing the non-maximal cut, as observed in ref.~\cite{Pogel:2024sdi}.

\begin{examplenew}{2}
Let us illustrate the previous discussion on the example of the unequal-mass sunrise integral from examples \ref{exSunrise}, \ref{exSunrise2} and \ref{exSunrise3}. 
We can rotate the basis of master integrals by a constant matrix leading to a rotation of the intersection matrix
\begin{align}
\bK_N=\widetilde{\bR}^{-1}_{\mathrm{S},R}\bDelta_{0,\mathrm{S}}\widetilde{\bR}^{-1,T}_{\mathrm{S},R}\,, \text{~~~with~~~} \widetilde{\bR}^{-1}_{\mathrm{S},R} =\sqrt{2}\left(
\begin{array}{cccc}
 -i & -\frac{1}{2 \sqrt{3}} & 0 & -i \\
 -1 & 0 & -\frac{1}{2} & 1 \\
 -1 & 0 & \frac{1}{2} & 1 \\
 -i & \frac{1}{2 \sqrt{3}} & 0 & -i \\
\end{array}
\right)\,.
\label{eq:SR_Galois_rot1}
\end{align}
After this transformation the symmetry relation for the differential equation matrix reads
\begin{align}
\label{eq:selfDualitySunrise2}
\widetilde{\bA}_\mathrm{S}^{\text{m.c.}}(\bx)=\bK_N\widetilde{\bA}_\mathrm{S}^{\text{m.c.},T}(\bx)\bK_N \,, \text{~~~with~~~} \widetilde{\bA}_\mathrm{S}^{\text{m.c.}}(\bx)=\widetilde{\bR}^{-1}_{\mathrm{S},R}\bA_\mathrm{S}^{\text{m.c.}}(\bx)\widetilde{\bR}_{\mathrm{S},R}\,,
\end{align}
which is precisely the self-duality observed in ref.~\cite{Pogel:2024sdi}. We clearly see how the transformation into the persymmetric form introduces spurious square roots in $\widetilde{\bA}_\mathrm{S}^{\text{m.c.}}(\bx)$, which arise from diagonalising $\bDelta_{0,\mathrm{S}}$. Therefore, we need to pass from $\mathbb{Q}$ to a field extension $\mathbb{Q}(i,\sqrt{3})$.\footnote{The factor $\sqrt{2}$ in eq.~\eqref{eq:SR_Galois_rot1} cancels in eq.~\eqref{eq:selfDualitySunrise2}.} As observed in ref.~\cite{Pogel:2024sdi}, this transformation to the persymmetric form is not unique. Another possible choice to obtain $\bK_N=\widetilde{\bR}^{-1}_{\mathrm{S},R,2}\bDelta_{0,\mathrm{S}}\widetilde{\bR}^{-1,T}_{\mathrm{S},R,2}$ would be 
\begin{align}
\widetilde{\bR}_{\mathrm{S},R,2}^{-1}=\frac{1}{\sqrt{2}}\left(
\begin{array}{cccc}
 4 i  & 0 & 0 & 0 \\
 0 & \frac{i}{\sqrt{3}} & -1 & 0 \\
 0 & \frac{i}{\sqrt{3}} & \phantom{-}1 & 0 \\
 0 & 0 & 0 & 4 i  \\
\end{array}
\right)\,.
\end{align}
This transformation leads to a different symmetry for the differential equation matrix, which is exactly the one found in ref.~\cite{Pogel:2024sdi}.
There is however nothing particularly special about the form in eq.~\eqref{eq:selfDualitySunrise2} of the self-duality. Indeed, we could even rotate in a way such that the intersection matrix becomes the identity
\begin{align}
\mathds{1}=\widetilde{\bR}^{-1}_{\mathrm{S},\mathbb{I}}\bDelta_{0,\mathrm{S}}\widetilde{\bR}^{-1,T}_{\mathrm{S},\mathds{1}}\,, \text{~~~with~~~} \widetilde{\bR}^{-1}_{\mathrm{S},\mathds{1}}=\left(
\begin{array}{cccc}
 0 & -\frac{i}{\sqrt{3}} & 0 & 0 \\
 0 & 0 & -i & 0 \\
 -2 i & 0 & 0 & -2 i \\
 -2 & 0 & 0 & 2 \\
\end{array}
\right) \,.
\end{align}
The required field extension is again $\mathbb{Q}(i,\sqrt{3})$, just like in the previous case. However, we can easily transform to yet another basis where a different field extension is required:
For example, we can choose to rotate $\bDelta_{0,\mathrm{S}}$ into any real symmetric matrix. Let us illustrate this by the transformation
\begin{align}
\text{diag}(2,3,5,6)=\widetilde{\bR}^{-1}_{\mathrm{S},2356}\bDelta_{0,\mathrm{S}}\widetilde{\bR}^{-1,T}_{\mathrm{S},2356}\,, 
\end{align}
with
\begin{align}
 \widetilde{\bR}^{-1}_{\mathrm{S},2356}=\left(
\begin{array}{cccc}
 0 & -i \sqrt{\frac{2}{3}} & 0 & 0 \\
 0 & 0 & -i \sqrt{3} & 0 \\
 -2 i \sqrt{5} & 0 & 0 & -2 i \sqrt{5} \\
 -2 \sqrt{6} & 0 & 0 & 2 \sqrt{6} \\
\end{array}
\right) \,.
\end{align}
The required field extension is now $\mathbb{Q}(i,\sqrt{2},\sqrt{3},\sqrt{5})$. 
 We see that the square roots that are introduced depend on the choice of basis, cf. also ref.~\cite{Pogel:2024sdi}.
\end{examplenew}


\section{Conclusion}
\label{sec:conclusion}

The main purpose of this paper was to analyse the self-duality and the Galois symmetry for maximal cuts recently introduced in ref.~\cite{Pogel:2024sdi}, and to see if and how it connects to the natural self-duality that comes from twisted cohomology. 
The self-duality of ref.~\cite{Pogel:2024sdi} is formulated at the level of the differential equations for a canonical basis of maximal cuts, and the existence of such a canonical basis is itself conjectural in the general case. 

In order to be able to formulate rigorous mathematical proofs, we have identified a class of differential equations, dubbed in C-form (see Definition~\ref{def:canon}), which captures all known examples of differential equations in canonical form. On a technical level, we argued that to every differential equation in C-form we can associate the following data:
\begin{enumerate}
\item an algebra $\cA$ of functions,
\item a finite-dimensional vector space $\bbA\subseteq Z^1(\cA)$ of differential one-forms,
\item a Lie algebra $\mathfrak{g}_{\bbA}$, and its universal enveloping algebra
$A_{\bbA}$,
\item a representation $\rho$ of $\mathfrak{g}_{\bbA}$,
\end{enumerate}
We could then show that, if both the differential equation for the Feynman integrals and their duals are transformed into a canonical basis which is also in C-form (which, we repeat, is the case for all known examples), then the intersection matrix between twisted co-cycles and their duals is constant. This constant intersection matrix in turn expresses the fact that the representation $\check{\rho}$ of $\mathfrak{g}_{\bbA}$ obtained from the dual Feynman integrals is equivalent to the dual $\rho^*$of the representation $\rho$ obtained from the Feynman integrals. Our result shows that there is a connection between canonical bases, Lie algebra representations and intersection theory, which we believe is interesting in its own right and may deserve more exploration in the future.

We then focused on the case of maximal cuts, where the representation $\rho$ is irreducible and self-dual. We find that, when combined with the constancy of the intersection matrix, this implies constraints on the representation $\rho$, which can be translated into symmetry properties of the differential equation matrix for the maximal cuts. We showed that, up to a change of basis, we can reproduce the observation of self-duality of ref.~\cite{Pogel:2024sdi} for sectors with a number $N$ of master integrals with either $N=2$ or $N$ odd. If $N$ is even and at least four, we see that there may be an additional symmetry property which is not accounted for by the observations in ref.~\cite{Pogel:2024sdi}. Whether or not this additional symmetry arises for Feynman integrals remains to be seen, and may deserve further investigations.

\section*{Acknowledgements}
The work of CS is supported by the CRC 1639 ``NuMeriQS'', and the work of CD and FP is funded by the European Union
(ERC Consolidator Grant LoCoMotive 101043686). Views
and opinions expressed are however those of the author(s)
only and do not necessarily reflect those of the European
Union or the European Research Council. Neither the
European Union nor the granting authority can be held
responsible for them.

\appendix


\section{Proof of linear independence of iterated integrals}
\label{app:DDMS}
In this appendix we collect the details of the proofs omitted in section~\ref{sec:line_dep}. We start by reviewing how to find a basis for the algebra $A_{\bbA}$ in appendix~\ref{app:DDMS_1}, and we present the proof of Theorem~\ref{thm:lin_dep}  in appendix~\ref{app:DDMS_2}. 

\subsection{Quadratic algebras}
\label{app:DDMS_1}
In this appendix we review how to find a basis for the algebra $A_{\bbA}$ generated by the non-commutative variables $t_i$ that appear in the generating functional in eq.~\eqref{eq:genfunc_original}.

The algebra ${A}_{\mathbb{A}}$ can be explicitly described as follows. Define ${A}_{\mathbb{A},1}$ to be the $\KK$-vector space generated by the $t_i$. Then we have
\beq
{A}_{\mathbb{A}} = T\left( {A}_{\mathbb{A},1}\right)/I_{\mathbb{A}}\,,
\eeq
where $T\left( {A}_{\mathbb{A},1}\right)$ is the tensor algebra of $ {A}_{\mathbb{A},1}$ and $I_{\mathbb{A}}$ is the two-sided ideal generated by the relations in eq.~\eqref{eq:t_commutators}. Because the differential forms $\omega_i$ are closed, the relations in eq.~\eqref{eq:t_commutators} are homogeneous of degree two in the $t_i$. This identifies ${A}_{\mathbb{A}}$ as a \emph{quadratic algebra}, cf.~ref.~\cite{AIF_1987__37_4_191_0}. Note that ${A}_{\mathbb{A}}$ is graded by the length of the words in the letters $t_i$. 
It is known how to find a linear basis for a quadratic algebra. Let $\mathcal{G}$ be a Groebner basis for the ideal ${I}_{\mathbb{A}}$ for some monomial ordering $\prec$ that is compatible with the length of the words, i.e., if $w_1$ and $w_2$ are words of lengths $|w_1|< |w_2|$, then $w_1\prec w_2$. An example of such an ordering is simply the lexicographic ordering. We set 
\beq
\mathcal{B}_{\mathbb{A}} = \big\{w\in T\left( {A}_{\mathbb{A},1}\right): \lm(g)\nmid w\,,\textrm{ for all }g\in\mathcal{G}\big\}\,,
\eeq
where $\lm(g)$ denotes the leading monomial of $g$ for the chosen monomial ordering and $w_1\nmid w_2$ means that $w_1$ does not divide $w_2$, i.e., $w_2$ does not lie in the two-sided ideal generated by $w_1$. Then it can be shown that the set $\mathcal{B}_{\mathbb{A}}$ is a basis for $A_{\mathbb{A}}$ (cf.,~e.g.,~ref.~\cite{shelper}). Note that are algorithms to construct a Groebner basis, and thus a basis $\cB_{\bbA}$.

Having identified a basis $\cB_{\bbA}$, let us introduce some notations that will be useful in appendix~\ref{app:DDMS_2} to prove Theorem~\ref{thm:lin_dep}. We start by interpreting elements of $A_\bbA$ as kets, and a basis of $A_\bbA$ is given by the elements $|w\rangle$ with $w\in\cB_{\bbA}$. The elements of the dual $A_\bbA^*$ are bras with basis $\langle w|$ with $w\in\cB_{\bbA}$, and there is a scalar product 
\beq
\langle .|.\rangle : A_{\bbA}^* \otimes_\KK A_{\bbA} \to \KK;\quad w\otimes w' \mapsto \langle w|w'\rangle\,,
\eeq
and $\langle w|w'\rangle = \delta_{ww'}$ for $w,w'\in\cB_{\bbA}$. In other words, the elements of $\cB_{\bbA}$ define an orthonormal basis for this scalar product. There is also a completeness relation,
\beq
\sum_{w\in \cB_{\bbA}} |w\rangle\langle w| = \mathds{1}\,.
\eeq
We can use this completeness relation to decompose any element $|u\rangle\in A_{\bbA}$ into the basis
\beq
|u\rangle = \sum_{w\in \cB_{\bbA}} |w\rangle\langle w|u\rangle\,.
\eeq
If we introduce the shorthand $I(w') = I_{\gamma}(\omega_{i_1},\ldots, \omega_{i_k})$ for $w' = t_{i_1}\cdots t_{i_k}$, we can write the generating function from eq.~\eqref{eq:genfunc_original} in the form
\beq
|\mathbb{G}\rangle = \sum_{w'} |w'\rangle\, I(w') = \sum_{w\in \cB_{\bbA}} |w\rangle\,J(w)\,,
\eeq
with 
\beq
J(w) = \sum_{w'} \langle w|w'\rangle\, I(w')\,. 
\eeq
Note that $\langle w|w'\rangle=0$ unless $|w|=|w'|$. 

\subsection{Proof of Theorem~\ref{thm:lin_dep}}
\label{app:DDMS_2}
We now present the proof of Theorem~\ref{thm:lin_dep}, which is essentially identical to the proof presented in ref.~\cite{deneufchatel:hal-00558773} for the case of a single variable.

\underline{$1.\Rightarrow 2.$:} Consider a finite set $\cS\subset \cB_{\bbA}$. We want to show that a relation of the type $\sum_{w\in \cS} f_w\, J(w)=0$ with $f_w\in\cF$ 
can only be satisfied for $f_w=0$. Equivalently, consider the dual vector
\beq
\langle P| = \sum_{w\in \cS} f_w\, \langle w|\,,
\eeq
and the map 
\beq
\phi_{\mathbb{G}}: \cF\otimes_\KK A_{\mathbb{A}}^* \to \cF \otimes_\KK V_{\bbA}; \quad \langle p| \mapsto \langle p|\mathbb{G}\rangle = \sum_{w\in \cS} f_w\, J(w)\,.
\eeq
Our goal is to show that $\Ker\phi_{\mathbb{G}}=\{0\}$. The strategy is to assume $\Ker\phi_{\mathbb{G}}\neq\{0\}$ and to show that this leads to a contradiction. Let $\langle P|$ be the non-zero element of $\Ker\phi_{\mathbb{G}}$ with the smallest value of $\lm(P)$ (where, as defined in the previous subsection, $\lm(P)$ the leading monomial for the chosen mononial ordering used to define the basis $\cB_{\bbA}$). If $\lm(P) = \langle w_0|$, we can write
\beq
\langle P| = f_{w_0}\langle w_0|+\sum_{\substack{w\in \cB_{\bbA}\\w\prec w_0}} f_w\langle w|\,,\qquad f_{w_0}\neq 0\,.
\eeq
Note that we must have $|w_0|>0$. Indeed, if $|w_0|=0$, then $w_0=1$ is the empty word, and we have $\phi_{\mathbb{G}}(P) = f_{w_0}$. Since $\langle P| \in \Ker\phi_{\mathbb{G}}$, this implies $f_{w_0}=0$, but then $\langle P|=0$, and we assumed that $\langle P|$ is non zero.

We define $\langle Q| := \tfrac{1}{f_{w_0}}\langle P|$. 
We differentiate $\phi_{\mathbb{G}}(Q) = \langle Q|\mathbb{G}\rangle=0$ and use eq.~\eqref{eq:gen_func_DEQ} . We obtain (note that the differential only acts on the coefficients, which are functions):
\beq\bsp
0=\rd \phi_{\mathbb{G}}(Q) &\,= \langle \rd Q|\mathbb{G}\rangle + \langle Q|\rd\mathbb{G}\rangle\\
&\,= \langle \rd Q|\mathbb{G}\rangle + \sum_{i=1}^p\omega_i\langle Q|t_i\mathbb{G}\rangle\\
&\,= \langle \rd Q|\mathbb{G}\rangle + \sum_{i=1}^p\omega_i\langle t_i^{\dagger}Q|\mathbb{G}\rangle\\
&\,= \langle \rd Q + \sum_{i=1}^p\omega_i t_i^{\dagger}Q|\mathbb{G}\rangle\,,
\esp\eeq
with 
\beq\label{eq:Q_exp}
\langle \rd Q + \sum_{i=1}^p\omega_i t_i^{\dagger}Q| = \sum_{i=1}^p\omega_i \langle t_i^{\dagger}w_0| + \sum_{\substack{w\in \cB_{\bbA}\\w\prec w_0}} \left[\rd f_w\langle w| + \sum_{i=1}^p\omega_i \langle t_i^{\dagger}w|\right]\,.
\eeq
Note that the action of the hermitian conjugate $t_i^{\dagger}$ lowers the length of a word. Indeed, for any words $w_1,w_2$, we have $\langle t_i^{\dagger}w_1|w_2\rangle = \langle w_1|t_iw_2\rangle = 0$ unless $|w_1| = |t_iw_2| = |w_2|+1$, and $| t_i^{\dagger}w_1| = |w_2| = |w_1|-1$. Since our monomial ordering is compatible with the length of the words, it follows that all monomials are less than $w_0$. Since $\langle P|$ was the non-zero element in the kernel with the smallest value of $\lm(P)$, we must then have\footnote{Strictly speaking, both sides of eq.~\eqref{eq:Q_exp} do not lie in $\Ker\phi_{\mathbb{G}}$, but in $Z^1(\cA)\otimes_\KK\Ker\phi_{\mathbb{G}}$. We can of course expand the differential forms into the basis $\rd x_i$, and we reach the same conclusion for each component. In the following we prefer to work with $Z^1(\cA)\otimes_\KK\Ker\phi_{\mathbb{G}}$ to keep the notations compact.}
\beq
\langle\rd Q| + \sum_{i=1}^p\omega_i \langle t_i^{\dagger}Q| = 0\,.
\eeq
Hence, for all words $w$ (not just those from $\cB_{\bbA}$), we must have
\beq\label{eq:proof_eq_1}
\rd\langle Q|w\rangle = - \sum_{i=1}^p\omega_i \langle t_i^{\dagger}Q|w\rangle = - \sum_{i=1}^p\omega_i \langle Q|t_iw\rangle\,.
\eeq

Assume now that $|w| = |w_0|$, and so $|t_iw| =|w|+1>|w_0|$. Since all monomials in $\langle Q|$ have length at most $|w_0|$ (because $w_0$ was the leading monomial, and hence the largest for our monomial ordering), we must have $\langle Q|t_iw\rangle =0$. Hence, eq.~\eqref{eq:proof_eq_1} implies
\beq\label{eq:app_proof_x}
\rd\langle Q|w\rangle = 0 \,,\qquad \textrm{for all } |w|=|w_0|\,.
\eeq

Since $|w_0|>0$, $w_0$ contains at least one letter $t_k$, and so we can write $\langle w_0 |= \langle t_kv|$, for some (possibly empty) word $v$. Note that we have $\langle Q|t_kv\rangle = \langle Q|w_0\rangle=1$. If we now let $w=v$ in eq.~\eqref{eq:proof_eq_1}, we find,
\beq\label{eq:proof_eq_2}
\rd\langle Q|v\rangle =   -\sum_{i=1}^p\omega_i\,\alpha_i\,, 
\eeq
with $\alpha_i :=  -\langle Q|t_iv\rangle$. Since $|t_iv| = |v|+1=|w_0|$, eq.~\eqref{eq:app_proof_x} implies that $\alpha_i\in \KK$ for all $1\le i\le p$. Moreover, $\langle Q|v\rangle \in \cF$. Hence, the left-hand side of eq.~\eqref{eq:proof_eq_2} lies in $\rd\cF$, while the right-hand side lies in $\bbA$. But, since $\rd\cF\cap \bbA=\{0\}$, this means
\beq\label{eq:proof_eq_3}
\rd\langle Q|v\rangle =   \sum_{i=1}^p\omega_i\,\alpha_i =0\,,
\eeq
and since the $\omega_i$ form a basis of $\bbA$, we must have $\alpha_i=0$, for all $1\le i\le p$. This is a contradiction, because we also have for $i=k$ that $\langle Q|t_kv\rangle = \langle Q|w_0\rangle=1$.
\vskip0.5cm

\underline{$2. \Rightarrow 1.$:} We need to show that the linear independence of the iterated integrals implies $\rd\cF\cap\bbA=\{0\}$. In other words, we need to show that, if there are $\alpha_1,\ldots,\alpha_p\in\KK$ and $f\in\cF$ such that
\beq\label{eq:ansatz}
\sum_{i=1}^p\alpha_i\omega_i = \rd f\,,
\eeq
then necessarily $\alpha_1=\ldots=\alpha_p=\rd f=0$. 

Assume that eq.~\eqref{eq:ansatz} holds.
Let $\langle P| = -f\langle 1| + \sum_{i=1}^p\alpha_i\langle t_i|$. We have
\beq
\langle P|\mathbb{G}\rangle = -f \langle 1|\mathbb{G}\rangle + \sum_{i=1}^p\alpha_i\langle t_i|\mathbb{G}\rangle = -f + \sum_{i=1}^p\alpha_i\, I_{\gamma}(\omega_i)\,,
\eeq
where we used the fact that $J(t_i) = I_{\gamma}(\omega_i)$. Equation~\eqref{eq:ansatz} implies $\rd\langle P|\mathbb{G}\rangle = 0$. Hence, there is a constant $\lambda \in \mathbb{C}$ such that $\langle P|\mathbb{G}\rangle=\lambda$. Let 
\beq
\langle Q| := \langle P| - \lambda \langle 1|\,.
\eeq
Then clearly
\beq
0 = \langle Q|\mathbb{G}\rangle = -(f+\lambda)\cdot 1 + \sum_{i=1}^p\alpha_i\, I_{\gamma}(\omega_i)\,.
\eeq
Since $1$ and $I_{\gamma}(\omega_i)$ are linearly independent over $\cF$ by hypothesis, we must have
\beq
\alpha_1=\ldots=\alpha_p=0 \textrm{~~~and~~~} f = -\lambda\,.
\eeq

\bibliographystyle{JHEP}
\bibliography{refs.bib}

\providecommand{\href}[2]{#2}\begingroup\raggedright\begin{thebibliography}{10}

\bibitem{THOOFT1972189}
G.~{'t Hooft} and M.~Veltman, {\it Regularization and renormalization of gauge fields},  {\em Nuclear Physics B} {\bf 44} (1972), no.~1 189--213.

\bibitem{Bourjaily:2022bwx}
J.~L. Bourjaily et~al., {\it {Functions Beyond Multiple Polylogarithms for Precision Collider Physics}},  in {\em {Snowmass 2021}}, 3, 2022.
\newblock \href{http://arxiv.org/abs/2203.07088}{{\tt 2203.07088}}.

\bibitem{Henn:2013pwa}
J.~M. Henn, {\it {Multiloop integrals in dimensional regularization made simple}},  {\em Phys. Rev. Lett.} {\bf 110} (2013) 251601, [\href{http://arxiv.org/abs/1304.1806}{{\tt 1304.1806}}].

\bibitem{Mastrolia:2018uzb}
P.~Mastrolia and S.~Mizera, {\it {Feynman Integrals and Intersection Theory}},  {\em JHEP} {\bf 02} (2019) 139, [\href{http://arxiv.org/abs/1810.03818}{{\tt 1810.03818}}].

\bibitem{Mizera:2017rqa}
S.~Mizera, {\it {Scattering Amplitudes from Intersection Theory}},  {\em Phys. Rev. Lett.} {\bf 120} (2018), no.~14 141602, [\href{http://arxiv.org/abs/1711.00469}{{\tt 1711.00469}}].

\bibitem{Mizera:2019gea}
S.~Mizera, {\em {Aspects of Scattering Amplitudes and Moduli Space Localization}}.
\newblock PhD thesis, Princeton, Inst. Advanced Study, 2020.
\newblock \href{http://arxiv.org/abs/1906.02099}{{\tt 1906.02099}}.

\bibitem{Mizera:2019ose}
S.~Mizera, {\it {Status of Intersection Theory and Feynman Integrals}},  {\em PoS} {\bf MA2019} (2019) 016, [\href{http://arxiv.org/abs/2002.10476}{{\tt 2002.10476}}].

\bibitem{matsumoto_relative_2019-1}
K.~Matsumoto, {\it Relative twisted homology and cohomology groups associated with {Lauricella}'s ${F}_d$},  Dec., 2019.

\bibitem{Abreu:2019wzk}
S.~Abreu, R.~Britto, C.~Duhr, E.~Gardi, and J.~Matthew, {\it {From positive geometries to a coaction on hypergeometric functions}},  {\em JHEP} {\bf 02} (2020) 122, [\href{http://arxiv.org/abs/1910.08358}{{\tt 1910.08358}}].

\bibitem{Britto:2021prf}
R.~Britto, S.~Mizera, C.~Rodriguez, and O.~Schlotterer, {\it {Coaction and double-copy properties of configuration-space integrals at genus zero}},  {\em JHEP} {\bf 05} (2021) 053, [\href{http://arxiv.org/abs/2102.06206}{{\tt 2102.06206}}].

\bibitem{Caron-Huot:2021xqj}
S.~Caron-Huot and A.~Pokraka, {\it {Duals of Feynman integrals. Part I. Differential equations}},  {\em JHEP} {\bf 12} (2021) 045, [\href{http://arxiv.org/abs/2104.06898}{{\tt 2104.06898}}].

\bibitem{Caron-Huot:2021iev}
S.~Caron-Huot and A.~Pokraka, {\it {Duals of Feynman Integrals. Part II. Generalized unitarity}},  {\em JHEP} {\bf 04} (2022) 078, [\href{http://arxiv.org/abs/2112.00055}{{\tt 2112.00055}}].

\bibitem{Chen:2022lzr}
J.~Chen, X.~Jiang, C.~Ma, X.~Xu, and L.~L. Yang, {\it {Baikov representations, intersection theory, and canonical Feynman integrals}},  {\em JHEP} {\bf 07} (2022) 066, [\href{http://arxiv.org/abs/2202.08127}{{\tt 2202.08127}}].

\bibitem{Giroux:2022wav}
M.~Giroux and A.~Pokraka, {\it {Loop-by-loop differential equations for dual (elliptic) Feynman integrals}},  {\em JHEP} {\bf 03} (2023) 155, [\href{http://arxiv.org/abs/2210.09898}{{\tt 2210.09898}}].

\bibitem{Gasparotto:2023roh}
F.~Gasparotto, S.~Weinzierl, and X.~Xu, {\it {Real time lattice correlation functions from differential equations}},  {\em JHEP} {\bf 06} (2023) 128, [\href{http://arxiv.org/abs/2305.05447}{{\tt 2305.05447}}].

\bibitem{Fontana:2023amt}
G.~Fontana and T.~Peraro, {\it {Reduction to master integrals via intersection numbers and polynomial expansions}},  {\em JHEP} {\bf 08} (2023) 175, [\href{http://arxiv.org/abs/2304.14336}{{\tt 2304.14336}}].

\bibitem{Bhardwaj:2023vvm}
R.~Bhardwaj, A.~Pokraka, L.~Ren, and C.~Rodriguez, {\it {A double copy from twisted (co)homology at genus one}},  {\em JHEP} {\bf 07} (2024) 040, [\href{http://arxiv.org/abs/2312.02148}{{\tt 2312.02148}}].

\bibitem{De:2023xue}
S.~De and A.~Pokraka, {\it {Cosmology meets cohomology}},  {\em JHEP} {\bf 03} (2024) 156, [\href{http://arxiv.org/abs/2308.03753}{{\tt 2308.03753}}].

\bibitem{Chen:2023kgw}
J.~Chen, B.~Feng, and L.~L. Yang, {\it {Intersection theory rules symbology}},  {\em Sci. China Phys. Mech. Astron.} {\bf 67} (2024), no.~2 221011, [\href{http://arxiv.org/abs/2305.01283}{{\tt 2305.01283}}].

\bibitem{Duhr:2023bku}
C.~Duhr and F.~Porkert, {\it {Feynman integrals in two dimensions and single-valued hypergeometric functions}},  {\em JHEP} {\bf 02} (2024) 179, [\href{http://arxiv.org/abs/2309.12772}{{\tt 2309.12772}}].

\bibitem{Brunello:2023rpq}
G.~Brunello, V.~Chestnov, G.~Crisanti, H.~Frellesvig, M.~K. Mandal, and P.~Mastrolia, {\it {Intersection Numbers, Polynomial Division and Relative Cohomology}},  \href{http://arxiv.org/abs/2401.01897}{{\tt 2401.01897}}.

\bibitem{Crisanti:2024onv}
G.~Crisanti and S.~Smith, {\it {Feynman Integral Reductions by Intersection Theory with Orthogonal Bases and Closed Formulae}},  \href{http://arxiv.org/abs/2405.18178}{{\tt 2405.18178}}.

\bibitem{Duhr:2024rxe}
C.~Duhr, F.~Porkert, C.~Semper, and S.~F. Stawinski, {\it {Twisted Riemann bilinear relations and Feynman integrals}},  \href{http://arxiv.org/abs/2407.17175}{{\tt 2407.17175}}.

\bibitem{Pogel:2022ken}
S.~P\"ogel, X.~Wang, and S.~Weinzierl, {\it {Taming Calabi-Yau Feynman Integrals: The Four-Loop Equal-Mass Banana Integral}},  {\em Phys. Rev. Lett.} {\bf 130} (2023), no.~10 101601, [\href{http://arxiv.org/abs/2211.04292}{{\tt 2211.04292}}].

\bibitem{Pogel:2022vat}
S.~P\"ogel, X.~Wang, and S.~Weinzierl, {\it {Bananas of equal mass: any loop, any order in the dimensional regularisation parameter}},  {\em JHEP} {\bf 04} (2023) 117, [\href{http://arxiv.org/abs/2212.08908}{{\tt 2212.08908}}].

\bibitem{Pogel:2022yat}
S.~P\"ogel, X.~Wang, and S.~Weinzierl, {\it {The three-loop equal-mass banana integral in \ensuremath{\varepsilon}-factorised form with meromorphic modular forms}},  {\em JHEP} {\bf 09} (2022) 062, [\href{http://arxiv.org/abs/2207.12893}{{\tt 2207.12893}}].

\bibitem{Pogel:2024sdi}
S.~P\"ogel, X.~Wang, S.~Weinzierl, K.~Wu, and X.~Xu, {\it {Self-dualities and Galois symmetries in Feynman integrals}},  \href{http://arxiv.org/abs/2407.08799}{{\tt 2407.08799}}.

\bibitem{Tkachov:1981wb}
F.~V. Tkachov, {\it {A theorem on analytical calculability of 4-loop renormalization group functions}},  {\em Phys. Lett. B} {\bf 100} (1981) 65--68.

\bibitem{Chetyrkin:1981qh}
K.~G. Chetyrkin and F.~V. Tkachov, {\it {Integration by parts: The algorithm to calculate $\beta$-functions in 4 loops}},  {\em Nucl. Phys. B} {\bf 192} (1981) 159--204.

\bibitem{Smirnov:2010hn}
A.~V. Smirnov and A.~V. Petukhov, {\it {The Number of Master Integrals is Finite}},  {\em Lett. Math. Phys.} {\bf 97} (2011) 37--44, [\href{http://arxiv.org/abs/1004.4199}{{\tt 1004.4199}}].

\bibitem{Bitoun:2017nre}
T.~Bitoun, C.~Bogner, R.~P. Klausen, and E.~Panzer, {\it {Feynman integral relations from parametric annihilators}},  {\em Lett. Math. Phys.} {\bf 109} (2019), no.~3 497--564, [\href{http://arxiv.org/abs/1712.09215}{{\tt 1712.09215}}].

\bibitem{Kotikov:1990kg}
A.~V. Kotikov, {\it {Differential equations method: New technique for massive Feynman diagrams calculation}},  {\em Phys. Lett. B} {\bf 254} (1991) 158--164.

\bibitem{Kotikov:1991hm}
A.~V. Kotikov, {\it {Differential equations method: The Calculation of vertex type Feynman diagrams}},  {\em Phys. Lett. B} {\bf 259} (1991) 314--322.

\bibitem{Kotikov:1991pm}
A.~V. Kotikov, {\it {Differential equation method: The Calculation of N point Feynman diagrams}},  {\em Phys. Lett. B} {\bf 267} (1991) 123--127. [Erratum: Phys.Lett.B 295, 409--409 (1992)].

\bibitem{Gehrmann:1999as}
T.~Gehrmann and E.~Remiddi, {\it {Differential equations for two loop four point functions}},  {\em Nucl. Phys. B} {\bf 580} (2000) 485--518, [\href{http://arxiv.org/abs/hep-ph/9912329}{{\tt hep-ph/9912329}}].

\bibitem{ChenSymbol}
K.~T. Chen, {\it {Iterated path integrals}},  {\em Bull.\ Amer.\ Math.\ Soc.} {\bf 83} (1977) 831.

\bibitem{Adams:2018yfj}
L.~Adams and S.~Weinzierl, {\it {The $\varepsilon$-form of the differential equations for Feynman integrals in the elliptic case}},  {\em Phys. Lett. B} {\bf 781} (2018) 270--278, [\href{http://arxiv.org/abs/1802.05020}{{\tt 1802.05020}}].

\bibitem{Broedel:2018rwm}
J.~Broedel, C.~Duhr, F.~Dulat, B.~Penante, and L.~Tancredi, {\it {From modular forms to differential equations for Feynman integrals}},  in {\em {KMPB Conference}: {Elliptic Integrals, Elliptic Functions and Modular Forms in Quantum Field Theory}}, pp.~107--131, 2019.
\newblock \href{http://arxiv.org/abs/1807.00842}{{\tt 1807.00842}}.

\bibitem{Adams:2018bsn}
L.~Adams, E.~Chaubey, and S.~Weinzierl, {\it {Planar Double Box Integral for Top Pair Production with a Closed Top Loop to all orders in the Dimensional Regularization Parameter}},  {\em Phys. Rev. Lett.} {\bf 121} (2018), no.~14 142001, [\href{http://arxiv.org/abs/1804.11144}{{\tt 1804.11144}}].

\bibitem{Dlapa:2022wdu}
C.~Dlapa, J.~M. Henn, and F.~J. Wagner, {\it {An algorithmic approach to finding canonical differential equations for elliptic Feynman integrals}},  {\em JHEP} {\bf 08} (2023) 120, [\href{http://arxiv.org/abs/2211.16357}{{\tt 2211.16357}}].

\bibitem{Muller:2022gec}
H.~M\"uller and S.~Weinzierl, {\it {A Feynman integral depending on two elliptic curves}},  {\em JHEP} {\bf 07} (2022) 101, [\href{http://arxiv.org/abs/2205.04818}{{\tt 2205.04818}}].

\bibitem{Klemm:2024wtd}
A.~Klemm, C.~Nega, B.~Sauer, and J.~Plefka, {\it {Calabi-Yau periods for black hole scattering in classical general relativity}},  {\em Phys. Rev. D} {\bf 109} (2024), no.~12 124046, [\href{http://arxiv.org/abs/2401.07899}{{\tt 2401.07899}}].

\bibitem{Bogner:2019lfa}
C.~Bogner, S.~M\"uller-Stach, and S.~Weinzierl, {\it {The unequal mass sunrise integral expressed through iterated integrals on $\overline{\mathcal M}_{1,3}$}},  {\em Nucl. Phys. B} {\bf 954} (2020) 114991, [\href{http://arxiv.org/abs/1907.01251}{{\tt 1907.01251}}].

\bibitem{Gorges:2023zgv}
L.~G\"orges, C.~Nega, L.~Tancredi, and F.~J. Wagner, {\it {On a procedure to derive \ensuremath{\epsilon}-factorised differential equations beyond polylogarithms}},  {\em JHEP} {\bf 07} (2023) 206, [\href{http://arxiv.org/abs/2305.14090}{{\tt 2305.14090}}].

\bibitem{Henn:2014qga}
J.~M. Henn, {\it {Lectures on differential equations for Feynman integrals}},  {\em J. Phys. A} {\bf 48} (2015) 153001, [\href{http://arxiv.org/abs/1412.2296}{{\tt 1412.2296}}].

\bibitem{Lee:2014ioa}
R.~N. Lee, {\it {Reducing differential equations for multiloop master integrals}},  {\em JHEP} {\bf 04} (2015) 108, [\href{http://arxiv.org/abs/1411.0911}{{\tt 1411.0911}}].

\bibitem{Gituliar:2017vzm}
O.~Gituliar and V.~Magerya, {\it {Fuchsia: a tool for reducing differential equations for Feynman master integrals to epsilon form}},  {\em Comput. Phys. Commun.} {\bf 219} (2017) 329--338, [\href{http://arxiv.org/abs/1701.04269}{{\tt 1701.04269}}].

\bibitem{Meyer:2017joq}
C.~Meyer, {\it {Algorithmic transformation of multi-loop master integrals to a canonical basis with CANONICA}},  {\em Comput. Phys. Commun.} {\bf 222} (2018) 295--312, [\href{http://arxiv.org/abs/1705.06252}{{\tt 1705.06252}}].

\bibitem{Lee:2020zfb}
R.~N. Lee, {\it {Libra: A package for transformation of differential systems for multiloop integrals}},  {\em Comput. Phys. Commun.} {\bf 267} (2021) 108058, [\href{http://arxiv.org/abs/2012.00279}{{\tt 2012.00279}}].

\bibitem{Henn:2020lye}
J.~Henn, B.~Mistlberger, V.~A. Smirnov, and P.~Wasser, {\it {Constructing d-log integrands and computing master integrals for three-loop four-particle scattering}},  {\em JHEP} {\bf 04} (2020) 167, [\href{http://arxiv.org/abs/2002.09492}{{\tt 2002.09492}}].

\bibitem{Dlapa:2020cwj}
C.~Dlapa, J.~Henn, and K.~Yan, {\it {Deriving canonical differential equations for Feynman integrals from a single uniform weight integral}},  {\em JHEP} {\bf 05} (2020) 025, [\href{http://arxiv.org/abs/2002.02340}{{\tt 2002.02340}}].

\bibitem{Ahmed:2024tsg}
T.~Ahmed, E.~Chaubey, M.~Kaur, and S.~Maggio, {\it {Two-loop non-planar four-point topology with massive internal loop}},  {\em JHEP} {\bf 05} (2024) 064, [\href{http://arxiv.org/abs/2402.07311}{{\tt 2402.07311}}].

\bibitem{Britto:2024mna}
R.~Britto, C.~Duhr, H.~S. Hannesdottir, and S.~Mizera, {\it {Cutting-Edge Tools for Cutting Edges}},  2, 2024.
\newblock \href{http://arxiv.org/abs/2402.19415}{{\tt 2402.19415}}.

\bibitem{Anastasiou:2002yz}
C.~Anastasiou and K.~Melnikov, {\it {Higgs boson production at hadron colliders in NNLO QCD}},  {\em Nucl. Phys. B} {\bf 646} (2002) 220--256, [\href{http://arxiv.org/abs/hep-ph/0207004}{{\tt hep-ph/0207004}}].

\bibitem{Anastasiou:2003yy}
C.~Anastasiou, L.~J. Dixon, K.~Melnikov, and F.~Petriello, {\it {Dilepton rapidity distribution in the Drell-Yan process at NNLO in QCD}},  {\em Phys. Rev. Lett.} {\bf 91} (2003) 182002, [\href{http://arxiv.org/abs/hep-ph/0306192}{{\tt hep-ph/0306192}}].

\bibitem{BAIKOV1997347}
P.~Baikov, {\it Explicit solutions of the multi-loop integral recurrence relations and its application},  {\em Nuclear Instruments and Methods in Physics Research Section A: Accelerators, Spectrometers, Detectors and Associated Equipment} {\bf 389} (1997), no.~1 347--349. New Computing Techniques in Physics Research V.

\bibitem{LEE2010474}
R.~Lee, {\it Space--time dimensionality d as complex variable: Calculating loop integrals using dimensional recurrence relation and analytical properties with respect to d},  {\em Nuclear Physics B} {\bf 830} (2010), no.~3 474--492.

\bibitem{yoshida_hypergeometric_1997}
M.~Yoshida, {\em Hypergeometric {Functions}, {My} {Love}}, vol.~32 of {\em Aspects of {Mathematics}}.
\newblock Vieweg+Teubner Verlag, Wiesbaden, 1997.

\bibitem{aomoto_theory_2011}
K.~Aomoto and M.~Kita, {\em Theory of {Hypergeometric} {Functions}}.
\newblock Springer {Monographs} in {Mathematics}. Springer Japan, 2011.

\bibitem{kita_intersection_1994-2}
M.~Kita and M.~Yoshida, {\it Intersection {Theory} for {Twisted} {Cycles}},  {\em Mathematische Nachrichten} {\bf 166} (1994), no.~1 287--304.

\bibitem{Chestnov:2022alh}
V.~Chestnov, F.~Gasparotto, M.~K. Mandal, P.~Mastrolia, S.~J. Matsubara-Heo, H.~J. Munch, and N.~Takayama, {\it {Macaulay matrix for Feynman integrals: linear relations and intersection numbers}},  {\em JHEP} {\bf 09} (2022) 187, [\href{http://arxiv.org/abs/2204.12983}{{\tt 2204.12983}}].

\bibitem{Chestnov:2022okt}
V.~Chestnov, {\it {Recent progress in intersection theory for Feynman integrals decomposition.}},  {\em PoS} {\bf LL2022} (2022) 058, [\href{http://arxiv.org/abs/2209.01464}{{\tt 2209.01464}}].

\bibitem{Munch:2022ouq}
H.~J. Munch, {\it {Feynman Integral Relations from GKZ Hypergeometric Systems}},  {\em PoS} {\bf LL2022} (2022) 042, [\href{http://arxiv.org/abs/2207.09780}{{\tt 2207.09780}}].

\bibitem{Cho_Matsumoto_1995}
K.~Cho and K.~Matsumoto, {\it Intersection theory for twisted cohomologies and twisted riemann's period relations i},  {\em Nagoya Mathematical Journal} {\bf 139} (1995) 67--86.

\bibitem{Zagier2008}
D.~Zagier, {\em Elliptic Modular Forms and Their Applications}, pp.~1--103.
\newblock Springer Berlin Heidelberg, Berlin, Heidelberg, 2008.

\bibitem{Giroux:2024yxu}
M.~Giroux, A.~Pokraka, F.~Porkert, and Y.~Sohnle, {\it {The soaring kite: a tale of two punctured tori}},  {\em JHEP} {\bf 05} (2024) 239, [\href{http://arxiv.org/abs/2401.14307}{{\tt 2401.14307}}].

\bibitem{10.1093/imrn/rnaa060}
P.~Lochak, N.~Matthes, and L.~Schneps, {\it {Elliptic Multizetas and the Elliptic Double Shuffle Relations}},  {\em International Mathematics Research Notices} {\bf 2021} (04, 2020) 695--753.

\bibitem{10.2140/ant.2017.11.2113}
N.~Matthes, {\it {On the algebraic structure of iterated integrals of quasimodular forms}},  {\em Algebra \& Number Theory} {\bf 11} (2017), no.~9 2113 -- 2130.

\bibitem{Broedel:2018iwv}
J.~Broedel, C.~Duhr, F.~Dulat, B.~Penante, and L.~Tancredi, {\it {Elliptic symbol calculus: from elliptic polylogarithms to iterated integrals of Eisenstein series}},  {\em JHEP} {\bf 08} (2018) 014, [\href{http://arxiv.org/abs/1803.10256}{{\tt 1803.10256}}].

\bibitem{Broedel_2018}
J.~Broedel, C.~Duhr, F.~Dulat, and L.~Tancredi, {\it Elliptic polylogarithms and iterated integrals on elliptic curves. part i: general formalism},  {\em Journal of High Energy Physics} {\bf 2018} (May, 2018).

\bibitem{zerbini_enriquez}
B.~Enriquez and F.~Zerbini, {\it {Elliptic Hyperlogarithms}},  \href{http://arxiv.org/abs/2307.01833}{{\tt 2307.01833}}.

\bibitem{deneufchatel:hal-00558773}
M.~Deneufch{\^a}tel, G.~H. Duchamp, V.~Hoang Ngoc~Minh, and A.~I. Solomon, {\it {Independence of hyperlogarithms over function fields via algebraic combinatorics.}},  in {\em {Lecture Notes in Computer Science, Volume 6742 (2011),}}.
\newblock May, 2011.

\bibitem{Remiddi:1999ew}
E.~Remiddi and J.~A.~M. Vermaseren, {\it {Harmonic polylogarithms}},  {\em Int. J. Mod. Phys. A} {\bf 15} (2000) 725--754, [\href{http://arxiv.org/abs/hep-ph/9905237}{{\tt hep-ph/9905237}}].

\bibitem{AIF_1987__37_4_191_0}
Y.~I. Manin, {\it Some remarks on {Koszul} algebras and quantum groups},  {\em Annales de l'Institut Fourier} {\bf 37} (1987), no.~4 191--205.

\bibitem{shelper}
A.~Shepler and S.~Witherspoon, {\it Poincare-birkhoff-witt theorems},  \href{http://arxiv.org/abs/1404.6497}{{\tt 1404.6497}}.

\end{thebibliography}\endgroup

\end{document}